\documentclass[10pt]{book}
\usepackage[sectionbib,round,authoryear]{natbib}
\usepackage{array,epsfig,fancyheadings,rotating}
\usepackage[dvipdfm]{hyperref}

\usepackage{amsmath}
\usepackage{amssymb}
\usepackage{amsfonts}
\usepackage{multirow}
\usepackage{amsthm}

\newtheorem{theorem}{Theorem}
\newtheorem{lemma}{Lemma}

\theoremstyle{definition}

\newtheorem{remark}{Remark}
\newcommand{\N}{\mathbb{N}}
\newcommand{\R}{\mathbb{R}}

\newcommand{\F}{\mathcal{F}}

\begin{document}


\renewcommand{\baselinestretch}{1.2}


\markboth{\hfill{\footnotesize\rm M. Gonz\'alez, C. Minuesa, and I. del Puerto} \hfill}
{\hfill {\footnotesize\rm Minimum disparity estimation in CBP} \hfill}

\renewcommand{\thefootnote}{}
$\ $\par


\fontsize{10.95}{14pt plus.8pt minus .6pt}\selectfont
\vspace{0.8pc}
\centerline{\large\bf Minimum disparity estimation in controlled branching processes}
\vspace{2pt}
\vspace{.4cm}
\centerline{Miguel Gonz\'alez, Carmen Minuesa*\footnote{*Corresponding author. Phone 0034924289300. Ext. 86820} and In\'es del Puerto}
\vspace{.4cm}
\centerline{\it University of Extremadura}
\vspace{.55cm}
\fontsize{9}{11.5pt plus.8pt minus .6pt}\selectfont


\begin{quotation}
\noindent {\it Abstract:}
Minimum disparity estimation in controlled branching processes is dealt with by assuming that the offspring  law belongs to a general parametric family. Under some regularity conditions it is proved that the minimum disparity estimators proposed -based on the nonparametric maximum likelihood estimator of the offspring law when the entire family tree is observed- are consistent and asymptotic normally distributed. Moreover, it is discussed the robustness of the estimators proposed. Through a simulated example, focussing on the minimum Hellinger and negative exponential disparity estimators, it is shown that both are robust against outliers, being the negative exponential one also robust against inliers.
\par

\vspace{9pt}
\par
\end{quotation}\par

\def\thefigure{\arabic{figure}}
\def\thetable{\arabic{table}}

\fontsize{10.95}{14pt plus.8pt minus .6pt}\selectfont

\setcounter{chapter}{1}
\setcounter{equation}{0} 
\noindent {\bf 1. Introduction}

Branching processes are useful models for the description of the dynamics of systems whose elements produce new ones following  probability laws. Its theory has been  developed from simple models to increasing realism. Added to the theoretical interest in these processes there is therefore a major practical dimension due to their potential applications in such diverse fields as biology, epidemiology, genetics, medicine, nuclear physics, demography, actuarial mathematics, algorithm and data structures, see, for example, the monographs \citet*{Devroye}, \citet*{Haccou}, \citet*{procedings-workshop} and \citet*{Kimmel}.

In particular, controlled branching processes  (CBPs) are discrete time sto\-chas\-tic processes very appropriate to describe the growth of populations in which the number of participating individuals in the reproduction process is determined in each generation by a control mechanism. Besides, as is common in the branching framework, every individual reproduces independently of the others following the same probability law, which is called the offspring distribution.

The versatility of these models makes possible to model different kind of migratory movements. Thus, several well-known branching processes can be included in this class as particular cases, for instance, the own Bienaym\'e--Galton--Watson process, the branching processes with immigration (see \citet*{Sriram-94} and references therein), with random migration (see \citet*{yanevyanev}), with immigration at state zero (see \citet*{art-Bruss}) or with bounded emigration (see \citet*{libro-ejem}). Other interesting particular cases are branching processes with adaptive control (see \citet*{Bercu}) or with continuous state space (see \citet*{ra}).

Since the appearance of the pioneering publication by  \citet*{Yanev-75}, its probability theory has been widely studied. The development of its inference theory, which guarantees the applicability of these models, has become  the main goal in the most recent researches. For this issue, in a frequentist framework, it is important to mention \citet*{de}, \citet*{Mohan-2000}, \citep*{art-2004b,art-2005b}, \citet*{Sriram}, and \citet*{art-EM}. From a Bayesian standpoint, one can find the papers \citet*{mmp} and \citet*{art-ABC}. It is well-known that, in general, the estimators based on maximizing the likelihood function are badly affected by outliers. This also happens in CBP context, as is pointed out in the simulated example at the end of the paper, reason why robust procedures must be developed.

Robust estimation has barely developed in the context of branching processes. One only can find  results in the frame of the Bienaym\'e--Galton--Watson processes, by using weighted least trimmed estimation (see \citet*{asy}) and by considering minimum Hellinger distance estimation (see \citet*{Sriram-2000}). Our main aim in this paper is to carry  out, in a frequentist framework, robust estimation for the general class of CBPs. To this end, we make use of the minimum disparity methodology. This methodology  has  arisen as one that attains robustness properties without loss of efficiency.  It was introduced in \citet*{lindsay} for discrete models and since then, the literature on it has experimented a large growth (see \citet*{Pardo-2006} and \citet*{Basu-2011} for further information). In our context, assuming that the offspring distribution belongs to a very general parametric family, we  determine minimum disparity estimators (MDEs) of the underlying parameter and study their asymptotic and robustness properties. The method consists of proposing as estimator of the offspring distribution that parametric distribution which minimizes the discrepancy with a nonparametric estimator based on some observed sample.  The discrepancy is measured by a function called disparity measure. Thus, one can obtain different MDEs depending on the nonparametric estimator and the disparity measure considered. Special interest is highlighted in this paper for the negative exponential disparity and the Hellinger distance. The maximum likelihood estimator based on the observation of the whole family tree until a certain generation is considered as the nonparametric estimator. This paper presents for the first time the application of the technique of minimum disparity for the general class of branching structure given by CBPs, hence extending the results in \citet*{Sriram-2000} in a double sense: model and measure, and moreover extending the results in   \citet*{Basu-MNEDE} and \citet*{Park-Basu-2004} from an independent and identically distributed (i.i.d.) and continuous context to a dependent and discrete setup. It is worthwhile to point out the fundamental roles played by the nonparametric estimator and the dependence structure of the CBP to obtain the asymptotic properties of the MDEs proposed, that require a different approach from those already established in the i.i.d. setting.

Besides the introduction, this paper is organized into 6 sections and an appendix. In Section  2 
 we present the formal model and establish some hypotheses that we assume throughout the paper. Section 3 
  is devoted to defining and describing minimum disparity estimation.  The asymptotic properties of MDEs are also studied; to this end, we introduce the disparity functional associated to a disparity measure and research its properties. The robustness of MDEs is studied in Section 4. 
   To illustrate this methodology, we present a simulated example in Section 5. 
    Concluding remarks about the contributions of the paper are presented in Section 6. 
    Finally, we dedicate an appendix  to the proofs of the theorems, in order to facilitate the reading of the paper.
\par

\setcounter{chapter}{2}
\setcounter{equation}{0} 
\noindent {\bf 2. The probability model}
\par
\label{sec:model}
We  consider a \emph{controlled branching process with random control function} (CBP). Mathematically, this process is a discrete--time stochastic model $\{Z_n\}_{n\in\N}$ defined recursively as:
\begin{equation}\label{def:model}
Z_0=N,\quad Z_{n+1}=\sum_{j=1}^{\phi_n(Z_{n})}X_{nj},\quad n=0,1,\ldots,
\end{equation}
being $N$ a nonnegative integer, $\{X_{nj}:\ n=0,1,\ldots;j=1,2,\ldots\}$ and $\{\phi_n(k):n,k=0,1,\ldots\}$ two independent families of nonnegative integer valued random variables. Moreover, $X_{nj}$, $n=0,1,\ldots$, $j=1,2,\ldots$, are i.i.d. random variables and for each $n=0,1,\ldots$, $\{\phi_n(k)\}_{k\geq 0}$, are independent stochastic processes with equal one--dimensional probability distributions. The empty sum in \eqref{def:model} is considered to be 0. We denote by $p=\{p_k\}_{k\geq 0}$ the common probability distribution of the random variables $X_{nj}$, i.e., $p_k=P[X_{nj}=k]$, $k\geq 0$, which is known as offspring distribution or reproduction law, and by $m$ and $\sigma^2$ its mean and variance (assumed finite), and we referred to them as offspring mean and variance, respectively. We also denote $\varepsilon(k)=E[\phi_0(k)]$ and $\sigma^2(k)=Var[\phi_0(k)]$ the mean and the variance of the control variables (assumed finite too).

In addition, we suppose that the offspring distribution belongs to a general parametric family:
\begin{equation}\label{family}
\mathcal{F}_\Theta=\{p(\theta): \theta\in\Theta\},
\end{equation}
where $p(\theta)=\{p_k(\theta)\}_{k\geq 0}$ and $\Theta$ is a subset of $\R$, that is, $p=p(\theta_0)$ for some $\theta_0\in\Theta$, referred as to the offspring parameter. It is the aim of this paper to estimate $\theta_0$ efficiently and robustly by choosing $\theta \in \Theta$ which provides the best adjustment to the observed sample in terms of the disparity measures.

To develop this methodology we need to consider nonparametric estimators of the offspring distribution. In this sense, in \citet*{art-EM}, nonparametric maximum likelihood estimators (MLEs) based on different samples are provided. Let denote a generic nonparametric estimator of $p$ based on a sample, say $\mathcal{X}_n$,  by $\tilde{p}_{n}=\{\tilde{p}_{n,k}\}_{k\geq 0}$, verifying $\tilde{p}_{n,k}\geq 0$, for each $k\geq 0$, and $\sum_{k=0}^\infty \tilde{p}_{n,k}=1$ (where $n$ indicates that we observe the data up to the generation $n$).

\setcounter{chapter}{3}
\setcounter{equation}{0} 
\noindent {\bf 3. Minimum disparity estimation}
\par
\label{sec:MDE}

In this section, we introduce the notions of disparity measure and minimum disparity estimator, and present several interesting examples of them. Although we focuss our attention on probability distributions defined on the nonnegative integers, that is, those which can be offspring distributions, the definitions and results given in this section keep valid for whatever discrete model. Let $\Gamma$ be the set of all probability distributions defined on the nonnegative integers,  $\mathcal{F}_\Theta$ the parametric family introduced in (\ref{family}), and $G(\cdot)$ a three times differentiable and strictly convex function on $[-1,\infty)$ with $G(0)=0$. The \emph{disparity measure} $\rho$ corresponding to $G(\cdot)$ is defined for any $q\in\Gamma$ and $\theta\in\Theta$, as
\begin{equation*}\label{def:disparity-measure}
    \rho(q,\theta)=\sum_{k=0}^\infty G(\delta(q,\theta,k))p_k(\theta),
\end{equation*}
where $\delta(q,\theta,k)$ denotes the \emph{``Pearson residual at $k$''}, that is,
\begin{equation*}\label{Pearson-res}
    \delta(q,\theta,k)=\frac{q_k}{p_k(\theta)}-1.
\end{equation*}
Notice that the Pearson residual at $k$ depends on the probability distribution $q$ and on the parameter $\theta$, and that $\delta(q,\theta,k)\in [-1,\infty)$, for each $q\in\Gamma$, $\theta\in\Theta$, and $k\geq 0$.

Due to the fact that $G(\cdot)$ is strictly convex, one has that $\rho$ is nonnegative. Moreover, when $G(\cdot)$ is also nonnegative and has a unique zero at $0$ it is verified that $\rho(q,\theta)=0$ if and only if $q=p(\theta)$. Given a sample $\mathcal{X}_n$ and a nonparametric estimator of $p$, $\tilde{p}_{n}$, based on it, we define the \emph{minimum disparity estimator (MDE) of $\theta_0$} for the disparity measure $\rho$ based on $\tilde{p}_{n}$ as
\begin{equation}\label{def:MDE}
    \tilde{\theta}^\rho_n(\tilde{p}_{n})=\arg\min_{\theta\in\Theta} \rho(\tilde{p}_{n},\theta).
\end{equation}

\begin{remark} Some interesting cases of nonnegative disparity measures are the following:

\noindent (a) The disparity obtained with the function $G(\delta)=(\delta+1)\log(\delta+1)$ is a kind of the Kullback--Leibler divergence. It is denoted by $LD(\tilde{p}_{n},\theta)$, for each $\tilde{p}_{n}$, $n\in\N$, and $\theta\in\Theta$, and it is known as \emph{likelihood disparity}. Its minimizer, $\tilde{\theta}^{LD}_n(\tilde{p}_{n})$, is known as the minimum likelihood disparity estimator (MLDE). In some cases, this estimator coincides with the MLE.

\noindent (b) The disparity determined by the function $G(\delta)=[(\delta+1)^{1/2}-1]^2$ is the \emph{squared Hellinger distance}, denoted by $HD(\tilde{p}_{n},\theta)$, for each $\tilde{p}_{n}$, $n\in\N$, and $\theta\in\Theta$. In this case, the MDE is the minimum Hellinger distance estimator (MHDE), denoted by $\tilde{\theta}^{HD}_n(\tilde{p}_{n})$.
   \item [(c)] The disparities defined by using either the function $G(\delta)=\exp(-\delta)-1$ or $G(\delta)=\exp(-\delta)-2$ (denoted by $D(\tilde{p}_{n},\theta)$ and $D_M(\tilde{p}_{n},\theta)$, respectively, for each $\tilde{p}_{n}$, $n\in\N$, and $\theta\in\Theta$) are  both known as \emph{negative exponential disparity} (notice both disparities differ only in a constant). The MDE is denoted by $\tilde{\theta}^{NED}_n(\tilde{p}_{n})$ and it is called the minimum negative exponential disparity estimator (MNEDE).
\end{remark}
Other examples are the \emph{family of power divergence measures} (see \citet*{cressie-84}), the \emph{blended chi-squared measures}, which include \emph{Pearson's chi--squared} and \emph{Neyman's chi-square}, the \emph{blended weight chi--squared measures} and the \emph{blended weight Hellinger distance family} (see \citet*{lindsay}).

Under conditions of differentiability of the model, a useful way for determining a MDE is to take into account that it must satisfy $\dot{\rho}(\tilde{p}_{n},\tilde{\theta}^\rho_n(\tilde{p}_{n}))=0$, with $\dot{\rho}(q,\theta)$ denoting the first derivative of $\rho(q,\theta)$ with respect to $\theta$, $q\in \Gamma$. It is verified, for $q\in \Gamma$ and $\theta\in \Theta$,
$$-\dot{\rho}(q,{\theta})=\sum_{k=0}^\infty p'_k(\theta)A(\delta(q,\theta,k)),$$
with $A(\delta)=(\delta+1)G'(\delta)-G(\delta)$, and $G'(\cdot)$ and $p'_k(\cdot)$ denoting the first derivative of $G(\cdot)$ and $p_k(\cdot)$, respectively. The function $A(\cdot)$ is called the \emph{residual adjustment function} (RAF) of the disparity and it is an increasing function on $[-1,\infty)$, and twice differentiable that can be redefined so that $A(0)=0$ and $A'(0)=1$.

\begin{remark}
\begin{enumerate}
  \item [(a)] The RAF of the likelihood disparity is $A(\delta)=\delta$.
  \item [(b)] The RAF of the squared Hellinger distance after the standardization is $A(\delta)=2[(\delta+1)^{1/2}-1]$.
  \item [(c)] The RAF of the negative exponential disparity corresponding to $G(\delta)=\exp(-\delta)-1$ is $A(\delta)=1-(2+\delta)\exp(-\delta)$ and to the function $G(\delta)=\exp(-\delta)-2$ is $A(\delta)=2-(2+\delta)\exp(-\delta)$.
\end{enumerate}
\end{remark}

In \citet*{lindsay}, the RAFs of different disparity measures are compared (see Figures 4 and 5 in \citet*{lindsay}). The RAF of a disparity measure is relevant in determining the efficiency and robustness properties of the corresponding MDE. Concretely, $A''(0)$ is demonstrated to play a key role: large negative values of $A''(0)$ correspond to robustness properties and zero value matches a second--order efficient estimator in the sense of \citet*{Rao}.


Before focussing on these matters, first we establish  the existence of the minimum in \eqref{def:MDE} and its uniqueness. To this end, we consider the \emph{disparity functional} associated with a disparity $\rho$, defined as $T^\rho: \Gamma \rightarrow \Theta$, with $T^\rho(q)=\arg\min_{\theta\in\Theta} \rho(q,\theta)$, whenever the minimum exists. Notice that there might exist multiple values of the parameter $\theta$ which minimize the function $\rho(q,\cdot)$. As a consequence, $T^\rho(q)$ would denote any of these values. Moreover, $\tilde{\theta}^\rho_n(\tilde{p}_{n})=T^\rho(\tilde{p}_{n})$.

It is easy to show that if the paramater space $\Theta$ is compact, then the disparity functional is well defined for each $q\in\Gamma$ such that the function $\rho(q,\cdot)$ is continuous in $\Theta$. However, we will weaken the compactness of $\Theta$  in a similar way as was done in \citet*{Simpson-1987} and in \citet*{Sriram-2000}. Specifically, given a disparity $\rho$, we limit our study to the subclass $\tilde{\Gamma}_\rho\subseteq\Gamma$ which verifies the following condition: there exists a compact set $C_\rho\subseteq\Theta$ such that for every $q\in\tilde{\Gamma}_\rho$,
 \begin{equation}\label{eq:cond-gamma-tilde}
 \inf_{\theta\in\Theta\backslash C_\rho} \rho(q,\theta) > \rho(q,\theta^*),
 \end{equation}
for some $\theta^*\in C_\rho$.

\begin{theorem}\label{thm:exist-uniq-T}
It is satisfied:
\begin{enumerate}
  \item [(i)] For each $q\in\tilde{\Gamma}_\rho$ verifying that $\rho(q,\cdot)$ is continuous in $C_\rho$, there exists $T^\rho(q)$.
  \item [(ii)] If $\rho$ is a disparity measure  and $\theta^*\in\Theta$ verifying $\inf_{\theta\in \Theta\backslash K}\rho(p(\theta^*),\theta)>0$ for some compact set $K\subseteq\Theta$ and $\rho(p(\theta^*),\cdot)$ is continuous in $K$, then $T^\rho(p(\theta^*))$ exists. Moreover, if $\F_\Theta$ is identifiable, that is, $p(\cdot)$ is injective, and the disparity $\rho$ can be redefined (without changing its minimizer) so that the related function $G(\cdot)$ is nonnegative and has a unique zero at 0, then $\theta^*=T^\rho(p(\theta^*))$.
\end{enumerate}
\end{theorem}

The proof is provided in Appendix.
\par

\begin{remark}\label{rem:continuity}
(i) Notice that the continuity of $p_k(\cdot)$ for each $k\geq 0$ in an arbitrary set $B\subseteq\Theta$ leads to the continuity in $B$ of the function $\rho(q,\cdot)$ associated with any disparity measure $\rho$ determined by a bounded function $G(\cdot)$, for each $q\in\Gamma$. This is deduced by a generalized dominated convergence theorem (see \citet*{Royden}, p.92). The aforesaid condition is satisfied by the negative exponential disparity. Although the Hellinger distance is defined by a nonbounded function $G(\cdot)$, in this case the condition of continuity of $p_k(\cdot)$ in $B$ for each $k\geq 0$ is enough to obtain the continuity of $HD(q,\cdot)$ in $B$ for each $q\in\Gamma$. This latter is followed by the Cauchy-Schwarz inequality and the Scheff\'e's theorem.

\noindent (ii) The redefinition of some disparities, without affecting their minimizer, so that the related functions $G(\cdot)$ will be nonnegative and have a unique zero at 0 is possible. For instance, for the negative exponential disparity we can consider the function $\bar{G}(\delta)=G(\delta)+\delta$ instead of $G(\delta)$, which verifies the previous properties and for each $q\in\Gamma$, $\sum_{k=0}^\infty\bar{G}(\delta(q,\theta,k)) p_k(\theta)=\sum_{k=0}^\infty G(\delta(q,\theta,k))p_k(\theta)$.
\end{remark}

\vspace*{3ex}
In order to study the asymptotic properties of the MDEs, we are to assume several conditions. Let fix the next assumptions:
\begin{itemize}
\item[(A1)] $\rho$  is a disparity measure associated with a function $G(\cdot)$ which verifies that $G(\cdot)$ and $G'(\cdot)$ are bounded in $[-1,\infty)$.
\item[(A2)] For every $p\in\mathcal{F}_\Theta,\  p_k(\cdot)$ is continuous in $C_{\rho}$ for each $k\geq 0$ (where $C_{\rho}$ is introduced in \eqref{eq:cond-gamma-tilde}).
\end{itemize}

Notice that under (A2), by Theorem \ref{thm:exist-uniq-T}, $T^{\rho}(q)$ exists for every $q\in\tilde{\Gamma}_{\rho}$. Let $\hat{\Gamma}_{\rho}$ be the set of $q\in \tilde{\Gamma}_{\rho}$ such that $T^{\rho}(q)$ is unique.
Now, in the following theorem the continuity of the disparity functional is established. Henceforth, all the limits are taken as $n\to\infty$.

\begin{theorem}\label{thm:contin-T}
Let $q$ and $\{q_n\}_{n\in\N}$ be in $\Gamma$ such that $q_n\to q$ in $l_1$. Assuming (A1), (A2) and that $q\in\hat{\Gamma}_{\rho}$, then $T^{\rho}(q_n)$ eventually exists and the functional $T^\rho(\cdot)$ is continuous in $q$, that is, $T^\rho(q_{n})\to T^\rho(q)$.
\end{theorem}

In the case of the Hellinger distance, despite condition (A1) is not satisfied by this disparity, one can establish the following result.

\begin{theorem}\label{thm:contin-T-HD}
Let $q$ and $\{q_n\}_{n\in\N}$ be in $\Gamma$ verifying  $||q_n^{1/2}-q^{1/2}||_2\to 0$ (where $||\cdot||_2$ denotes the $l_2$-norm defined on $\Gamma$ and for each $q\in\Gamma$, $q^{1/2}=\{q_k^{1/2}\}_{k\geq 0}$). If (A2) holds and $q\in\hat{\Gamma}_{HD}$, then $T^{HD}(q_n)$ eventually exists and the functional $T^{HD}(\cdot)$ is continuous in $q$, that is, $T^{HD}(q_n)\to T^{HD}(q)$.
\end{theorem}

The proofs of Theorems \ref{thm:contin-T} and \ref{thm:contin-T-HD} are given in Appendix.

\vspace*{3ex}

Recall that $p=p(\theta_0)$ is the true reproduction law. Observe that under (A1), if (A2) is verified and $p\in\hat{\Gamma}_{\rho}$, one obtains $T^{\rho}(p)=\theta_0$ and for the case of Hellinger distance, dropping (A1), one also has $T^{HD}(p)=\theta_0$. Next theorem establishes the strong consistency of the MDEs.

\begin{theorem}\label{thm:consistency-MDE}
Assume (A2)  and $p\in\hat{\Gamma}_{\rho}$, for the corresponding disparity $\rho$. Then under conditions which guarantee that $\tilde{p}_{n,k}$ is a strongly consistent estimator of $p_k(\theta_0)$, for each $k\geq 0$, one has that:
\begin{enumerate}
\item [(i)]  $\tilde{\theta}_n^{\rho}(\tilde{p}_{n})$ eventually exists, is a random variable and  $\tilde{\theta}_n^{\rho}(\tilde{p}_{n})\rightarrow \theta_0$ almost surely (a.s.) if (A1) holds.
\item [(ii)] $\tilde{\theta}_n^{HD}(\tilde{p}_{n})$ eventually exists, is a random variable and $\tilde{\theta}_n^{HD}(\tilde{p}_{n})\rightarrow \theta_0$ a.s.
\end{enumerate}
\end{theorem}

The proof can be consulted in Appendix.\par

\setcounter{chapter}{4}
\setcounter{equation}{0} 
\noindent {\bf 4. Asymptotic efficiency}
\par
\label{sec:asymp-prop-MDE}
The results given in the previous section are general in the sense that the explicit expression of the nonparametric estimator is not required, and one only needs to know its properties, as for example, its consistency. However, to establish  the asymptotic efficiency of the MDEs, explicit formulas of the nonparametric estimators are needed. Then, to develop this section we come back to the CBP context.
In \citet*{art-EM}, we give nonparametric estimators of the offspring distribution under several sampling schemes. In particular, in a complete data context, we consider the entire family tree up to generation $n$ can be observed, that is, 
the sample $\mathcal{Z}_n^*=\{Z_l(k): 0\leq l\leq n-1; k\geq 0\}$, where $Z_l(k)=\sum_{i=1}^{\phi_l(Z_l)}I_{\{X_{li}=k\}}$, $0\leq l\leq n-1$, $k\geq 0$, with $I_B$ standing for the indicator function of the set $B$. Recall  that in a general setting $p=\{p_k\}_{k\geq 0}$ is the offspring distribution. The MLE of $p_k$, for each $k\geq 0$, (see \citet*{art-EM}),  is given by $\widehat{p}_n=\{\widehat{p}_{n,k}\}_{k\geq 0}$:
\begin{equation} \label{def:MLE-pk}
\widehat{p}_{n,k}=\frac{Y_{n-1}(k)}{\Delta_{n-1}},\qquad k\geq 0,\\
\end{equation}
where  $\Delta_l=\sum_{j=0}^l \phi_j(Z_j)$, and $Y_{l}(k)=\sum_{j=0}^{l} Z_j(k)$, $k\geq 0$, $0\leq l\leq n-1$.

It is proved that  $\widehat{p}_{n,k}$ is strongly consistent for $p_k$ on $\{Z_n\to\infty\}$, for each $k\geq 0$, (see Theorem 3.6 in \citet*{art-EM}), under the following assumption:

(A3) Let consider a CBP satisfying that:
\begin{itemize}
\item[(a)] There exists $\tau=\lim_{k\to\infty} \varepsilon(k)k^{-1}<\infty$ and the sequence $\{\sigma^2(k)k^{-1}\}_{k\geq 1}$ is bounded.
\item[(b)] $\tau_m=\tau m >1$ and $Z_0$ is large enough such that  $P[Z_n\rightarrow\infty]>0$.
\item[(c)] $\{Z_n(\tau_m)^{-n}\}_{n\in\N}$ converges a.s. to a finite random variable  $W$ such that $P[W>0]>0$.
\item[(d)] $\{W > 0\}=\{Z_n\to\infty\}$  a.s.
\end{itemize}

\begin{remark}
For CBPs verifying (a), sufficient conditions for (b), (c) and (d) hold are discussed in \citet*{art-EM}.
\end{remark}
As a consequence, under (A3), from Theorem \ref{thm:consistency-MDE}~\emph{(i)} and \emph{(ii)}, one obtains, respectively, that the estimators $\tilde{\theta}_n^{\rho}(\hat{p}_n)$ and $\tilde{\theta}_n^{HD}(\hat{p}_n)$ are strongly consistent on  $\{Z_n\to\infty\}$.

Now, we focus our attention on the asymptotic efficiency. To this end, we must consider additional conditions on the functions $p(\cdot)$. We assume from now on that  for each $k\geq 0$, $p_{k}(\theta)$ is twice continuously differentiable with respect to $\theta$ and:

(A4) For $\theta\in\Theta, \epsilon>0$ and for each $\theta^*\in(\theta-\epsilon,\theta+\epsilon)$
\begin{itemize}
\item[(a)] $ |p'_k(\theta^*) |< J_k(\theta),\ \forall k\geq 0,\quad  \sum_{k=0}^\infty J_k(\theta)<\infty$,
\item[(b)] $ |p''_k(\theta^*) |< L_k(\theta), \ \forall k\geq 0,\quad  \sum_{k=0}^\infty L_k(\theta)<\infty$,
\item[(c)]$ |u(\theta^*,k)^2 p_k(\theta^*) |< M_k(\theta), \ \forall k\geq 0,\quad\hspace*{-1ex}  \sum_{k=0}^\infty M_k(\theta)<\infty$, where $u(\theta,k)=(\log p_k(\theta))' = p'_k(\theta)/p_k(\theta)$.
\end{itemize}
(A5) $\rho$  is a disparity measure with  RAF $A(\cdot)$ verifying that $A(\delta),\ A'(\delta),\ A'(\delta)(1+\delta)$ and $A''(\delta)(1+\delta)$ are bounded functions on $\delta\in[-1,\infty)$.
\begin{remark}
Notice that for a disparity $\rho$ satisfying (A5), (A4) is a sufficient condition to guarantee that $\rho(q,\theta)$ can be twice differentiable with respect to $\theta$.
\end{remark}

It is easy to check that the negative exponential disparity satisfies (A5) but the Hellinger distance does not. Thus, to establish  the efficiency of MHDE instead of the previous hypotheses we will assume the following condition on $s(\theta)=p(\theta)^{1/2}$ (as usual, $s(\theta)=\{s_k(\theta)\}_{k\geq 0}$) in a similar way to that in \citet*{Beran-77}:

(A6) For $\theta\in int(\Theta)$ (that is, $\theta$ in the interior of $\Theta$), $s(\theta)$ is twice differentiable in $l_2$, that is, there exist $s'(\theta)\in l_2$ and $s''(\theta)\in l_2$ verifying that for every $\beta$ in a neighbourhood of zero
$$s_k(\theta+\beta) =  s_k(\theta)+\beta s_k'(\theta) +  \beta v_k(\beta),$$ $$ s_k'(\theta+\beta) = s_k'(\theta)+\beta s_k''(\theta) +  \beta w_k(\beta),$$ where $\sum_{k=0}^\infty v_k(\beta)^2\to 0\ \text{ and }\ \sum_{k=0}^\infty w_k(\beta)^2\to 0,\ \text{ as }\beta\to 0$.

Note that since the Fisher information is $I(\theta_0)=\sum_{k=0}^\infty u(\theta_0,k)^2 p_k(\theta_0)=4||s'(\theta_0)||_2^{2}$, either from (A4)~\emph{(c)} or from $s'(\theta_0)\in l_2$, $I(\theta_0)<\infty$ is obtained. In addition, observe that although conditions (A1) and (A5) seem to be quite restrictive, they are satisfied by a wide set of disparities (see \citet*{Park-Basu-2004}).

\begin{theorem}\label{thm:asymptotic-normality}
Let be a CBP satisfying (A3), with $p=p(\theta_0)$ its offspring distribution. Moreover, assume (A2) and $p\in\hat{\Gamma}_\rho$ (recall that in this case $T^{\rho}(p)=\theta_0$).
\begin{enumerate}
\item [(i)] If (A1), (A4), and (A5) hold, $s'(\theta_0)\in l_1$, and supposing that any sequence of estimators $\{\varphi_n\}_{n\in\N}$ converging to $\theta_0$ in probability satisfies
  \vspace*{-0.5ex}\begin{eqnarray}
  \sum_{k=0}^\infty |p''_k(\varphi_n)-p''_k(\theta_0)|&\xrightarrow{P}& 0,\label{eq:cond-conver-prob1}\\
  \sum_{k=0}^\infty |u(\varphi_n,k)^2 p_k(\varphi_n)-u(\theta_0,k)^2 p_k(\theta_0)|&\xrightarrow{P}&0,\label{eq:cond-conver-prob2}
  \end{eqnarray}\vspace*{-0.5ex}
then, it is verified:
\begin{equation}\label{eq:efficiency}
\Delta_{n-1}^{1/2}(\tilde{\theta}_n^\rho(\hat{p}_n)-\theta_0)\xrightarrow{d} N\left(0,I(\theta_0)^{-1}\right),
\end{equation}
where $\stackrel{P}{\rightarrow}$ denotes the convergence in probability and $\stackrel{d}{\rightarrow}$ represents the convergence in distribution with respect to the probability $P[\cdot| Z_n\to\infty]$.
\item[(ii)] For the Hellinger distance, \eqref{eq:efficiency} also holds under the assumptions (A6), $\sum_{k=0}^\infty s_k''(\theta_0)p_k^{1/2}<0$ and $\theta_0\in int(\Theta)$.
\end{enumerate}
\end{theorem}

The proof of the previous theorem is given in Appendix.

\begin{remark}
Besides the MDEs based on the whole family tree, one can determine the ones based on other samples. In \citet*{art-EM}, we also study the maximum likelihood estimation of the offspring distribution under incomplete sampling schemes, considering the random samples given by the number of individuals and progenitors in each generation, that is, $\overline{\mathcal{Z}}_n=\{Z_0,Z_{l+1},\phi_l(Z_l): l=0,\ldots,n-1\}$, and by only the generation sizes, that is, $\mathcal{Z}_n=\{Z_0,\ldots,Z_{n}\}$. The proposed estimators for the offspring distribution, based on $\overline{\mathcal{Z}}_n$ and $\mathcal{Z}_n$, are obtained by the Expectation--Maximization algorithm (EM algorithm). Making use of these estimators, one can obtain MDEs of $\theta_0$ based on $\overline{\mathcal{Z}}_n$ and $\mathcal{Z}_n$, respectively.\end{remark}

\setcounter{chapter}{5}
\setcounter{equation}{0} 
\noindent {\bf 5. Robustness}
\label{sec:robustness}

In this section, we address the issue of the robustness of the MDEs of $\theta$. For this purpose, we will study the behaviour of the corresponding disparity functional under contamination by considering the following model:
\begin{equation}\label{eq:mixture-model}
    p(\alpha,\theta,L)=(1-\alpha)p(\theta)+\alpha\eta_L,
\end{equation}
where $\alpha\in (0,1)$, $\theta\in\Theta$, $L\in\N_0$ and  $\eta_L$ is a point mass distribution at a nonnegative integer $L$. This model is called mixture model for gross errors at $L$ and it represents the simplest context of contamination. This approach was introduced in \citet*{Tukey} and consists of assuming the \emph{contaminated model} instead of the model distribution in order to explain or incorporate the outliers.

In the analysis of robustness of an estimator, an essential tool is the \emph{influence curve}, which for each disparity $\rho$ is a function of $L\in\N_0$ defined as
\begin{equation*}
\lim_{\alpha\to 0}\alpha^{-1}(T^\rho(p(\alpha,\theta,L))-T^\rho(p(\theta))).
\end{equation*}
Although the unboundedness of this function is an indicator of the misbehaviour of the MDEs of $\theta$ in presence of outliers, the influence curve can be very a deceptive measure of robustness (see \citet*{lindsay}). For this reason, we will also examine the $\alpha$-\emph{influence curves of $T^\rho(\cdot)$}, which are functions of $L\in\N_0$ defined as $\alpha^{-1}(T^\rho(p(\alpha,\theta,L))-T^\rho(p(\theta)))$, for each $\alpha\in (0,1)$.

Next theorem, whose proof can be read in Appendix, provides an expression for the influence curves and establishes conditions under which the disparity functional is robust at $p(\theta)$ against $100\alpha\%$ contamination by gross error at an arbitrary integer.

\begin{theorem}\label{thm:influence-curves}
Suppose the parameter space $\Theta$ is compact, (A2) holds (with $C_\rho=\Theta$), and $\F_\Theta$ is identifiable. For every $\alpha\in (0,1)$ and every $\theta\in\Theta$:
\begin{enumerate}
\item[(i)] Let $\rho$ be a disparity measure which can be redefined (without changing its minimizer) so that the related function $G(\cdot)$ is nonnegative, has a unique zero at 0 and verifies (A1) and (A5). For each $\alpha\in (0,1)$, $q\in\Gamma$ and $t\in\Theta$, define $\rho^*(\alpha,q,t)=\sum_{k=0}^\infty G^*(\delta(q,t,k))p_k(t)$, with $G^*(\delta)=G((1-\alpha)\delta)$. If (A4) holds, $T^{\rho}(p(\theta,\alpha,L))$ is unique for all $L$, and there exists a strictly increasing function $f$ such that $f(\rho^*(\alpha,p(\theta),t))=\rho((1-\alpha)p(\theta),t)$, for each $\alpha\in (0,1)$ and $t,\theta\in\Theta$; then
  \begin{enumerate}
   \item [(a)] $\lim_{L\to\infty} T^\rho(p(\theta,\alpha,L))=\theta$.
   \item [(b)] $T^\rho(p(\theta,\alpha,L))$ is a bounded and continuous function of $L$.
   \item [(c)] $\lim_{\alpha\to 0} \alpha^{-1}(T^\rho(p(\theta,\alpha,L))-\theta)=(I(\theta)p_L(\theta))^{-1}p_L'(\theta).$
 \end{enumerate}
\item[(ii)]   For the Hellinger distance, if $T^{HD}(p(\theta))\in int(\Theta)$, $\sum_{k=0}^\infty s_k''(\theta)p_k^{1/2}<0$, (A6) holds,  and $T^{HD}(p(\theta,\alpha,L))$ is unique for all $L$; then (i-a), (i-b) and (i-c) are also satisfied. 
\end{enumerate}
\end{theorem}

Observe that $p_L'(\theta)(I(\theta)p_L(\theta))^{-1}$ can be an unbounded function of $L$. Nevertheless, from Theorem \ref{thm:influence-curves}~\emph{(a)} and \emph{(b)}, we have that for every $\alpha\in (0,1)$, the  $\alpha$--influence curves are bounded continuous functions of $L$ verifying $\lim_{L\to\infty} \alpha^{-1}\linebreak\cdot(T^{\rho}(p(\theta,\alpha,L))-\theta)=0$, and $\lim_{L\to\infty} \alpha^{-1}(T^{HD}(p(\theta,\alpha,L))-\theta)=0$, respectively. Consequently, the associated disparity functionals are robust at $p(\theta)$ against $100\alpha\%$ contamination by gross error at an arbitrary integer $L$.

\vspace*{3ex}

Another important concept in the study of the robustness is the \emph{breakdown point}. The breakdown point of a disparity functional $T^\rho(\cdot)$ at $q\in\Gamma$ is given by:
\begin{equation*}
\alpha^*(T^\rho,q)=\inf\left\{\alpha\in(0,1): b(\alpha;T^\rho,q)=\infty\right\},
\end{equation*}
with $b(\alpha;T^\rho,q)=\sup\ \{|T^\rho((1-\alpha)q+\alpha\overline{q})-T^\rho(q)|:\overline{q}\in\Gamma\}$.
Intuitively, the breakdown point represents the smallest amount of contamination that can cause the estimator to take arbitrarily large values. Note that $b(\alpha;T^\rho,q)=\infty$ is equivalent to the existence of a sequence of probability distributions $\{q_n\}_{n\in\N}$ verifying $|T^\rho((1-\alpha)q+\alpha q_n)-T^\rho(q)|\to\infty$ and in that case, we say there is \emph{breakdown} in $T^\rho(\cdot)$ for a level of contamination equals $\alpha$. The sequence $\{q_n\}_{n\in \mathbb{N}}$ is called sequence of contaminating probability distributions. This fact is useful for the establishment of a lower bound for the breakdown point of the MDEs of $\theta$ in the following theorem.

\begin{theorem}\label{thm:breakdown-point-MDE}
\begin{enumerate}
\item [(i)] Assume that the contaminant sequence $\{q_n\}_{n\in\N}$, the distributions of the family $\mathcal{F}_\Theta$, and $\theta^*\in\Theta$ verify:
\begin{enumerate}
\item [(a)] $\sum_{k=0}^\infty \min\{p_k(\theta^*),q_{n,k}\}\to 0$.
\item [(b)] $\sum_{k=0}^\infty \min\{p_k(\theta),q_{n,k}\}\to 0$, uniformly for $\theta\in\Theta$ such that $|\theta|\leq c$, for any fixed $c\in\R$.
\item [(c)] $\sum_{k=0}^\infty \min\{p_k(\theta^*),p_k(\theta_n)\}\to 0$ if $|\theta_n|\to\infty$.
\item [(d)] $G(-1)$ and $\lim_{t\to\infty} G(t)/t$ are finite.
\end{enumerate}
Then, the asymptotic breakdown point of the MDE of $\theta^*$ is at least $1/2$.
\item [(ii)] Assume (A2), $q\in\hat\Gamma_{HD}$, and $T^{HD}(q)\in int(\Theta)$. Let $\varrho(q,p(\theta))= \linebreak\sum_{k=0}^\infty (q_k p_k(\theta))^{1/2}$,  $\hat{\varrho}=\max_{\theta\in\Theta} \varrho(q,p(\theta))$, $\varrho^*= \lim_{M\to\infty} \sup_{|\theta|>M}  \varrho(q,p(\theta))$ and $h_n=(1-\alpha)q+\alpha q_n$, $0<\alpha<1$, with $q_n\in \Gamma$, for every $n$. Assume that for each $n\geq 1$, $T^{HD}(h_n)$ exists and is unique. It holds that if $\alpha<(\hat{\varrho}-\varrho^*)^2/[1+(\hat{\varrho}-\varrho^*)^2]$, then there is no sequence $\{h_n\}_{n\in\N}$ of the form  defined above for which $\lim_{n\to\infty}|T^{HD}(h_n)-T^{HD}(q)|=\infty$.
\end{enumerate}
\end{theorem}

The proof of \emph{(i)} is analogous to those given in Theorem 4.1 in \citet*{Park-Basu-2004} replacing integrals with sums. The proof of \emph{(ii)} is exactly the same as Theorem 3 in \citet*{Simpson-1987} and it is omitted. In particular, since $p=p(\theta_0)$, then $\hat{\varrho}=1$ and as a consequence, the asymptotic breakdown point for HD is at least $1/2$ when $\varrho^*=0$, which usually holds.

\setcounter{chapter}{6}
\setcounter{equation}{0} 
\noindent {\bf 6. Simulated example}

\label{sec:examples}

Through a simulated example, we compare the behaviour of the MHDEs, MNEDEs and MLDEs based on the whole family tree under an uncontaminated model and under mixture models for gross errors. To this end, we have considered as initial model a CBP starting with one individual and Poisson distributions as offspring and control distributions. In particular, the offspring distribution is a Poisson distribution with the parameter $\theta_0=7$ and the variable $\phi_n(k)$ follows Poisson distribution with parameter  $\lambda k$, with $\lambda=0.3$, for each $k\geq 0$, $n\geq 0$. Therefore, the offspring mean and variance are $m=\sigma^2=7$, and $\tau_m=\theta_0\lambda=2.1$ (see (A3) for definition). The parameter $\tau_m$ is referred to as the asymptotic mean growth rate and is the threshold parameter of the model (see \citet*{art-2005a}). In practice, control functions, $\phi_n(k)$, following Poisson distributions of parameters $\lambda k$ are appropriate to describe an environment with expected immigration or emigration depending on $\lambda>1$ or $<1$. In our example we consider a model with expected emigration although supercritical ($\tau_m>1$).

First, we  show that in a contamination-free context, MHDEs and MNEDEs are as efficient as MLDEs. To this end, we have simulated 10 generations of $N=100$ CBPs following the previous model, and we have estimated the relative efficiency of $\tilde{\theta}_n^{NED}(\hat{p}_n)$ to $\tilde{\theta}_n^{HD}(\hat{p}_n)$, of $\tilde{\theta}_n^{HD}(\hat{p}_n)$ to $\tilde{\theta}_n^{LD}(\hat{p}_n)$ and of $\tilde{\theta}_n^{NED}(\hat{p}_n)$ to $\tilde{\theta}_n^{LD}(\hat{p}_n)$ in each generation by the ratios of these mean squared errors:
$$\frac{\text{MSE(}HD\text{)}}{\text{MSE(}NED\text{)}}, \ \frac{\text{MSE(}LD\text{)}}{\text{MSE(}HD\text{)}},\  \frac{\text{MSE(}LD\text{)}}{\text{MSE(}NED\text{)}},$$
where MSE($\rho$)$=N^{-1}\sum_{i=1}^{N} (\tilde{\theta}_i^{\rho}(\hat{p}_n)-\theta_0)^2$, $\rho\in$ $\{LD,$ $NED,$ $HD\}$,
with $i$ indicating the simulated process, for $i=1,\ldots,N$, and $n$ the generation, for $n=1,\ldots,10$. The evolution of these estimates is shown in Figure \ref{im:ej2-effic-grid-estimates} (first row -left), where one observes that as generations go up MNED and MHD procedures are shown as efficient as the MLD one. 

In a contaminated context, to illustrate and compare the accuracy of the estimates obtained by MHD and MNED methods, we have considered several different contaminated models for the offspring distribution in the aforementioned CBP. Specifically, we have contaminated the reproduction law according to the mixture model for gross errors, for $\alpha=0.05,$ $0.1,$ $0.15,$ $0.2,$ $0.25,$ $0.3,$ $0.35,$ $0.4,$ $0.45,$ $0.5$, and $L=0,1,\ldots,25$, obtaining 260 different contaminated CBPs.

For a generic CBP, the information given by a sample observed until a fixed generation $n$ depends on its asymptotic mean growth rate, $\tau_m$; being poorer when $\tau_m\approx 1$ than when $\tau_m>1$. This implies that to compare the behaviour of the different estimators for each one of the contaminated models (which have asymptotic  mean growth rates, called $\tau_m(\theta_0,\alpha, L)$, of different magnitudes) one needs to observe different numbers of generations depending on $\tau_m(\theta_0,\alpha, L)$- going in our example from $n=8$ for $\tau_m(\theta_0,\alpha, L)=4.8$ to $n=65$ for $\tau_m(\theta_0,\alpha, L)=1.05$.

\begin{figure*}
\scalebox{0.23}{\includegraphics[angle=270]{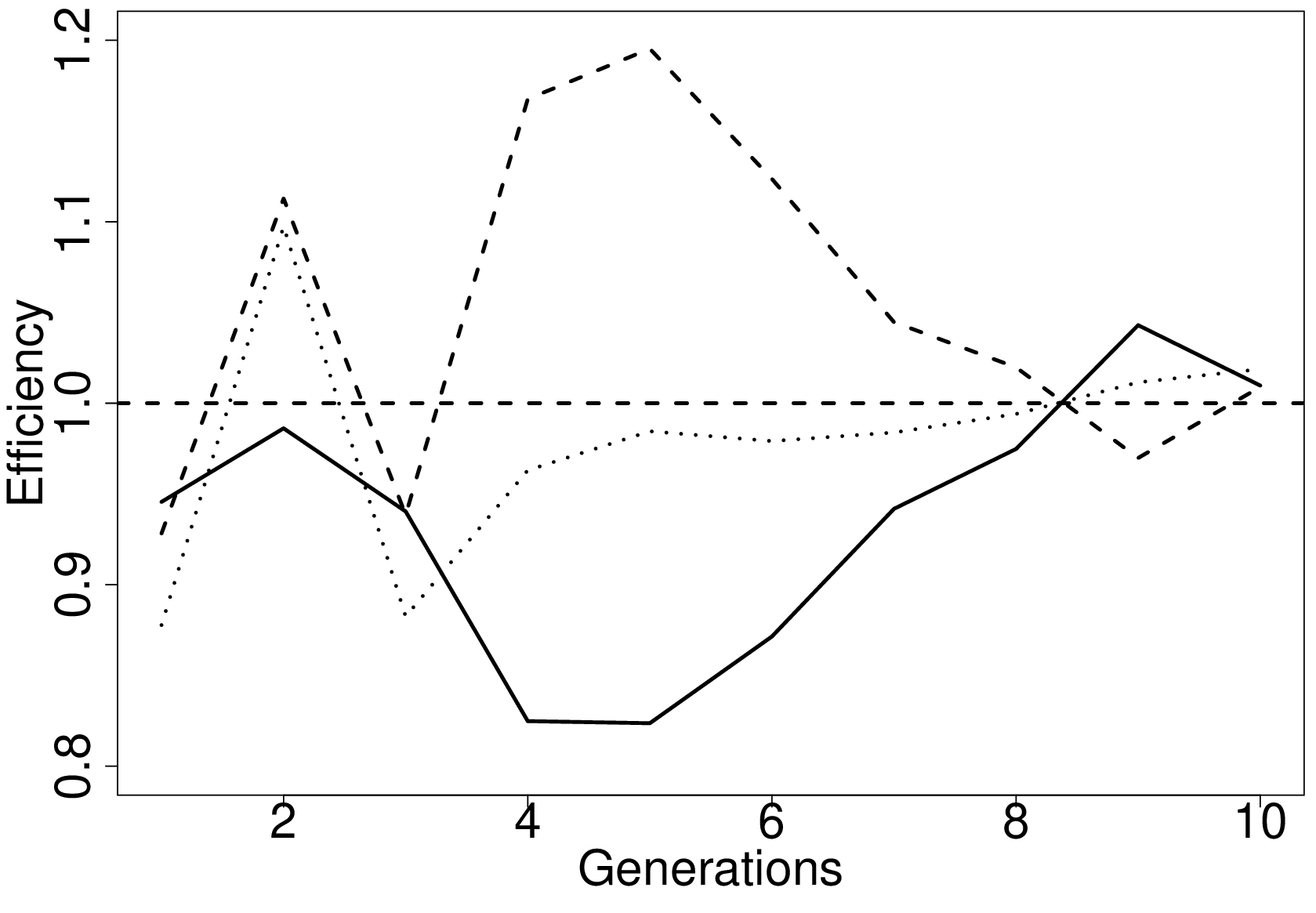}\hspace*{3em}
\includegraphics[angle=270]{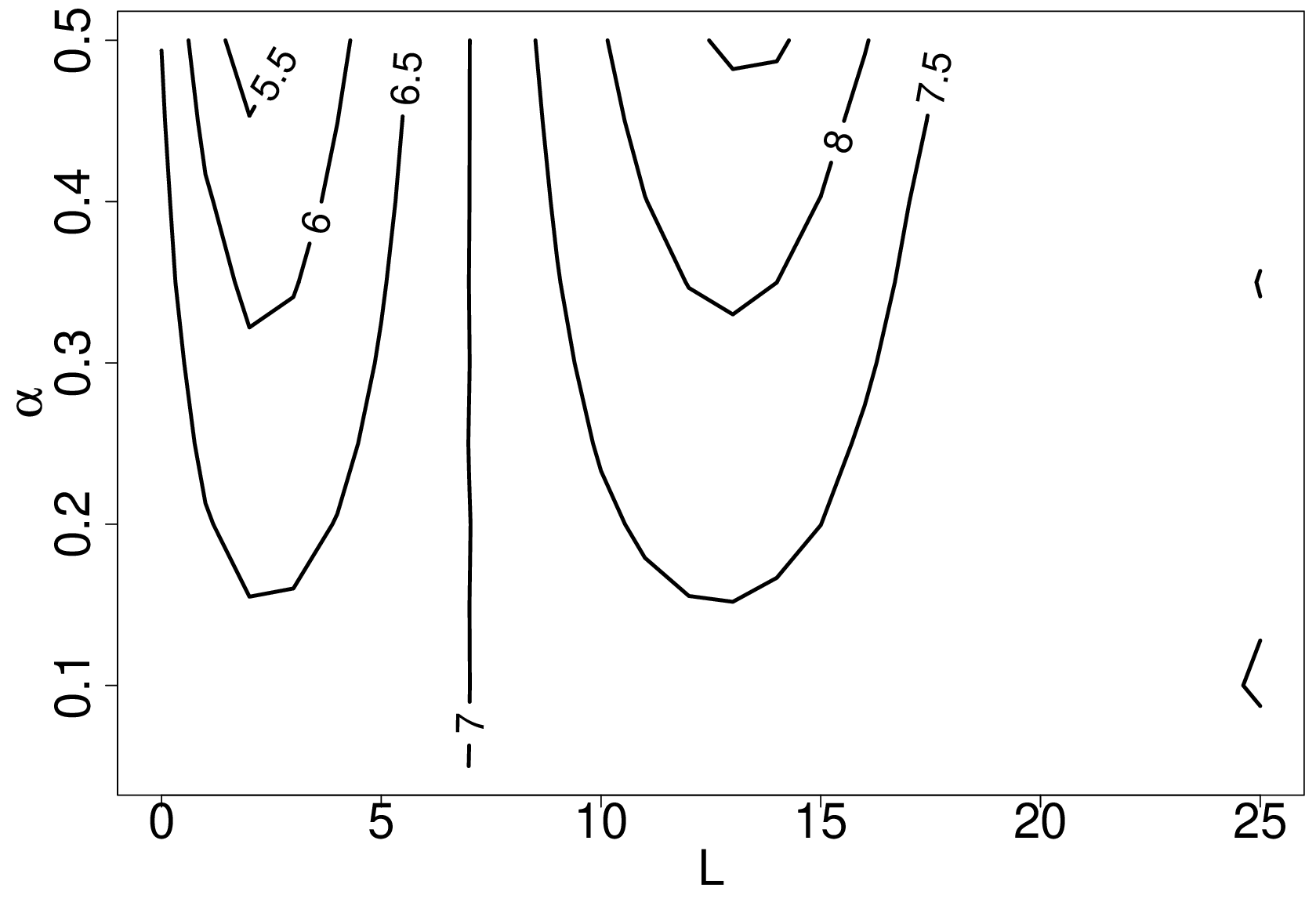}}\vspace*{0.2cm}

\scalebox{0.23}{\includegraphics[angle=270]{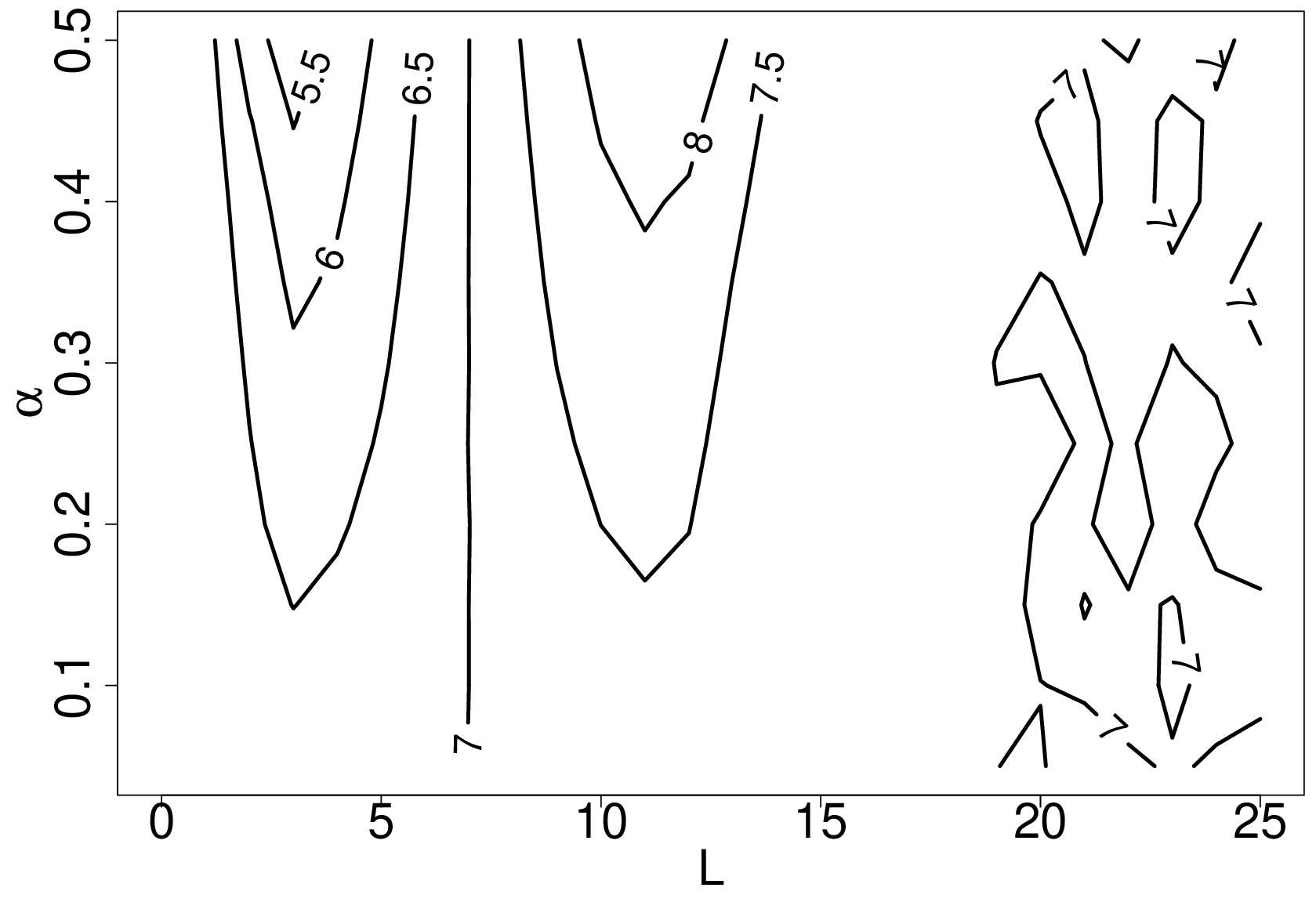}\hspace*{3em}
\includegraphics[angle=270]{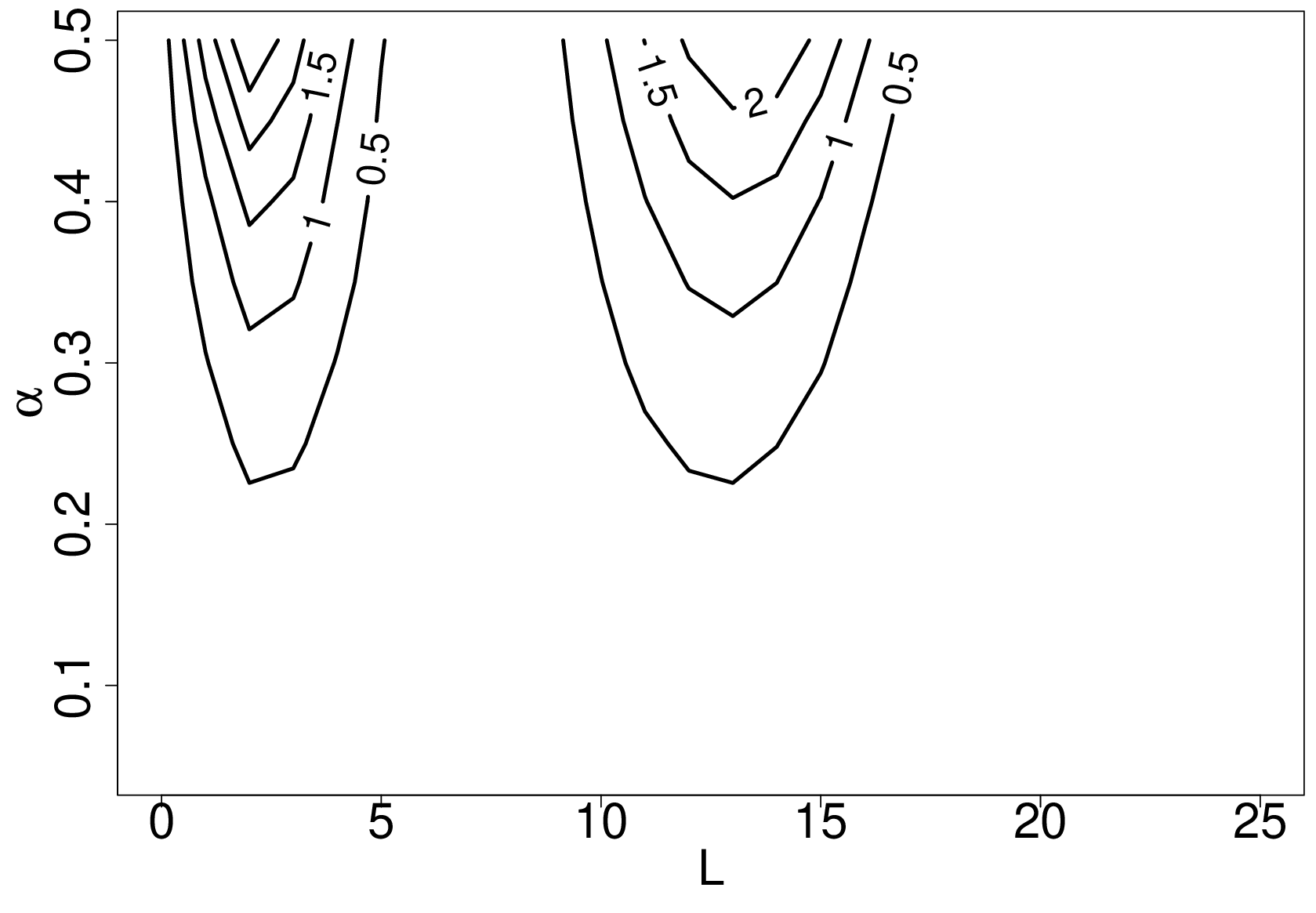}}\vspace*{0.2cm}

\scalebox{0.23}{\includegraphics[angle=270]{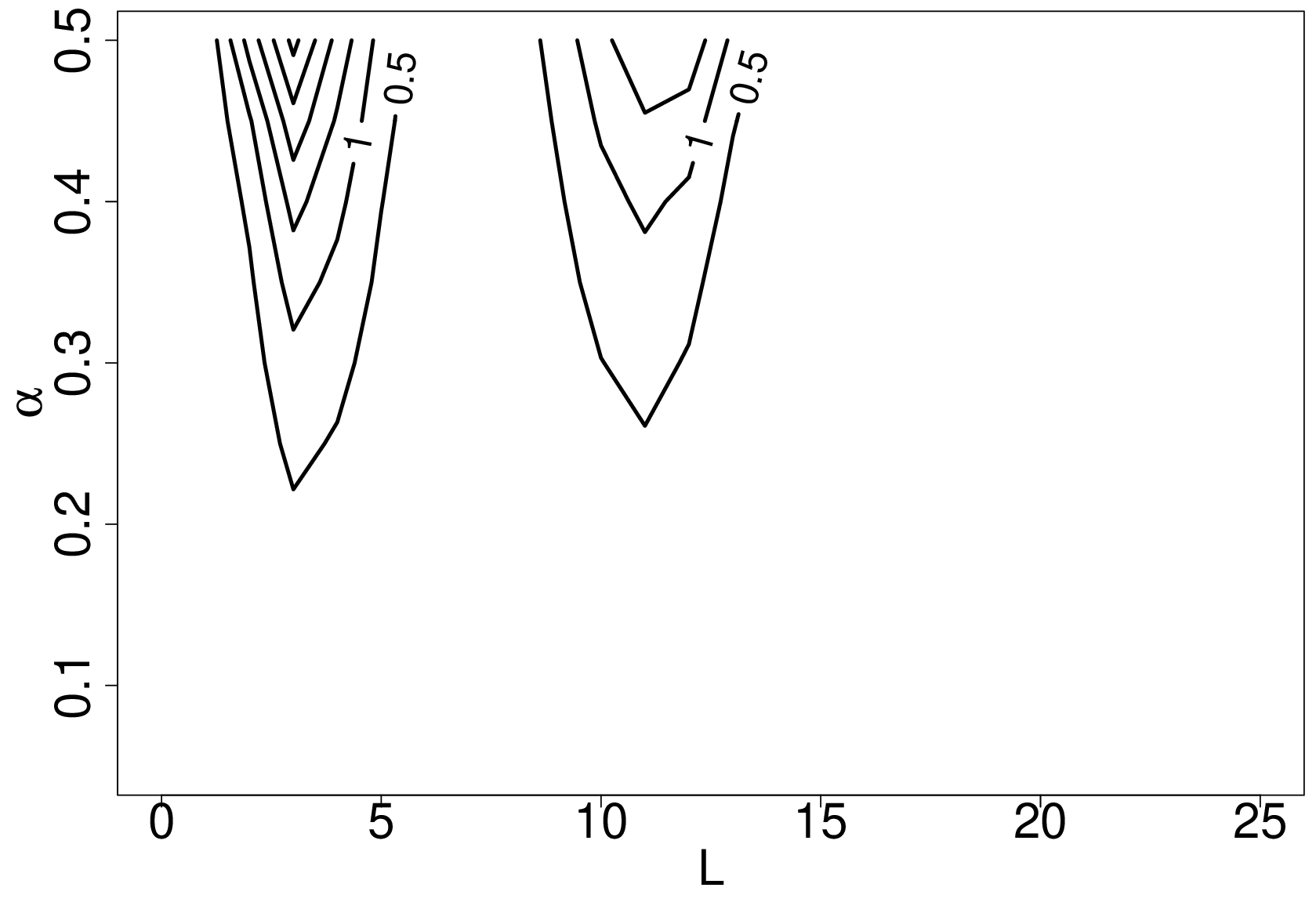}\hspace*{3em}
\includegraphics[angle=270]{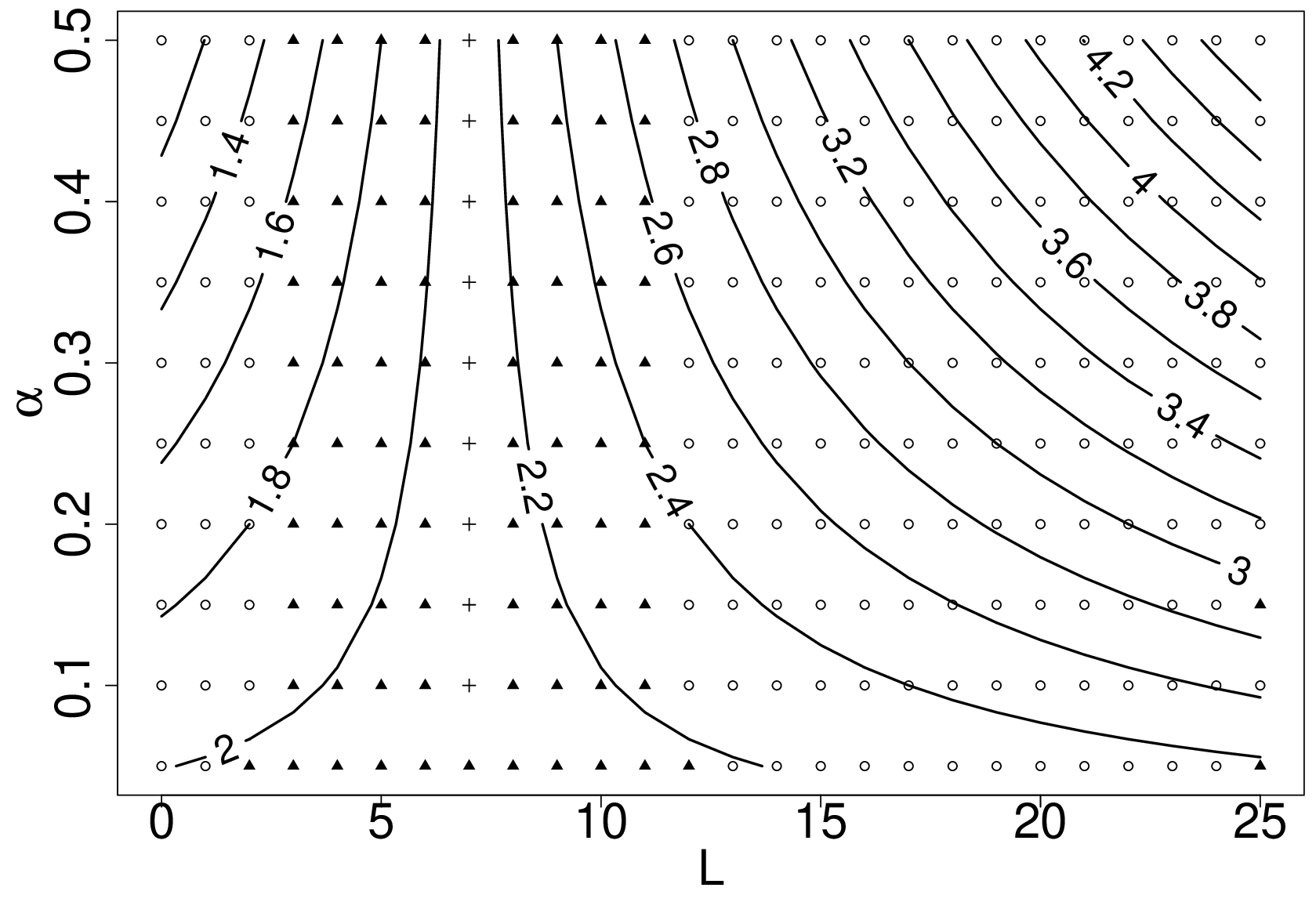}}
\caption{First row. Left: evolution of the estimates of the relative efficiency of $\tilde{\theta}^{NED}(\hat{p}_n)$ to $\tilde{\theta}^{HD}(\hat{p}_n)$ (solid line), the relative efficiency of $\tilde{\theta}^{HD}(\hat{p}_n)$ to $\tilde{\theta}^{LD}(\hat{p}_n)$ (dashed line) and the relative efficiency of $\tilde{\theta}^{NED}(\hat{p}_n)$ to $\tilde{\theta}^{LD}(\hat{p}_n)$ (dotted line). Right: contour plots of the means of the MHDEs. Second row. Left: the means of the MNEDEs of $\theta_0=7$ for each contaminated offspring distribution. Right: contour plots of the MSEs of the MHDEs of $\theta_0=7$. Thrird row. Left: MSEs of the MNEDEs of $\theta_0=7$  for each contaminated offspring distribution. Right: Contour plot of the asymptotic mean growth rates of the contaminated models (solid line) and points $(L,\alpha)$ where the minimum of MSE of the estimates of $\theta_0=7$ by the three methods is attained in the MLDE (crosses), in the MHDE (filled triangles) and in the MNEDEs (circles). }\label{im:ej2-effic-grid-estimates}
\end{figure*}

For each simulated process, we have determined the MHDEs, MNEDEs and MLDEs of $\theta_0$ in its last generation. In Figure \ref{im:ej2-effic-grid-estimates}, we show the mean (over the 100 simulations) of the MHDEs (first row -right) and of the MNEDEs (second row -left) of $\theta_0$, for each one of the 260 contaminated models. Moreover, the respective MSEs for both methods are represented in Figure \ref{im:ej2-effic-grid-estimates} (second row -right and third row -left).

  In addition, Figure \ref{im:ej2-effic-grid-estimates} (thrid row -right) shows the contour plot of the  asymptotic mean growth rate of the contaminated models,
  $\tau_m(\theta_0,\alpha,L)$, with  the underlying  points representing the minimum disparity method which provides the smallest MSE for each contaminated model. In view of these plots, one can deduce that the MNEDE supplies more accurate estimates in most of the contaminated models (166 models, that is $63.85\%$ of the models), but the best method when the contaminated state is between 3 and 11 is usually the MHD (85 models, that is $32.69\%$ of models). However, the MLDE only behaves properly in 9 models ($3.46\%$ of the models), where $L$ is equal 7 (consequently, the offspring mean remains unchanging).
\begin{table}
\caption{Relative bias for Hellinger distance and negative exponential disparity for the mixture models for gross errors with $L=0$ and different values of $\alpha$.}  \label{tab:ej2-inlier0}
\centering\begin{tabular}{c | c | c }
\hline\noalign{\smallskip}
$\alpha$ & $\frac{\Delta T^{HD}(\alpha,L)}{\Delta T^{LD}(\alpha,L)}$ & $\frac{\Delta T^{NED}(\alpha,L)}{\Delta T^{LD}(\alpha,L)}$  \\
\noalign{\smallskip}\hline\noalign{\smallskip}
$-0.0001$ & $1.0108310$ & $0.9980851$ \\
$-0.0002$ & $1.0519480$ & $0.9911580$ \\
$-0.0003$ & $1.1041350$ & $0.9787859$ \\
$-0.0004$ & $1.1383230$ & $0.9600473$ \\
$-0.0005$ & $1.1891900$ & $0.9338905$ \\
$-0.0006$ & $1.2520860$ & $0.9004372$ \\
$-0.0007$ & $1.3377720$ & $0.8556661$ \\
$-0.0008$ & $1.4711780$ & $0.7992881$ \\
$-0.0009$ & $1.7769600$ & $0.7296631$ \\
\noalign{\smallskip}\hline
\end{tabular}
\end{table}
We have also studied the performance of these methods in presence of inliers, which correspond to the model introduced at \eqref{eq:mixture-model} with $\alpha<0$ such that $p(\theta, \alpha, L)$ is a probability distribution. To this end, we compare the potential bias, defined as $\Delta T^{\rho} (\alpha, L)=T^{\rho}(p(\theta_0,\alpha,L))-T^{\rho}(p(\theta_0))$, with $\rho\in\{NED, HD, LD\}$. In fact,  we examine the relative bias of MHDE and MNEDE with respect to MLDE under mixture model for gross errors located at $L=0$ (Table \ref{tab:ej2-inlier0}), at $L=8$  (Table \ref{tab:ej2-inlier8}) and at $L=20$ (Table \ref{tab:ej2-inlier20}) for different values of $\alpha$. The results show that the MNEDE has decreasingly less bias than the MLDE in all the cases, whereas the inliers have the opposite effect on the MHDE.
\begin{table}
\caption{Relative bias for Hellinger distance and negative
exponential disparity for the mixture models for gross errors with
$L=8$ and different values of $\alpha$.}  \label{tab:ej2-inlier8}
\centering\begin{tabular}{c | c | c } \hline\noalign{\smallskip}
$\alpha$ & $\frac{\Delta T^{HD}(\alpha,L)}{\Delta T^{LD}(\alpha,L)}$ & $\frac{\Delta T^{NED}(\alpha,L)}{\Delta T^{LD}(\alpha,L)}$  \\
\noalign{\smallskip}\hline\noalign{\smallskip}
-0.01&1.022095&1.0004884\\
-0.02&1.041636&0.9985267\\
-0.03&1.064554&0.9947680\\
-0.04&1.089620&0.9895525\\
-0.05&1.117349&0.9830758\\
-0.06&1.148303&0.9737196\\
-0.07&1.183068&0.9622187\\
-0.08&1.222421&0.9480543\\
-0.09&1.267872&0.9306867\\
\noalign{\smallskip}\hline
\end{tabular}
\end{table}

\begin{table}
\caption{Relative bias for Hellinger distance and negative exponential disparity for the mixture models for gross errors with $L=20$ and different values of $\alpha$.} \label{tab:ej2-inlier20}
\centering\begin{tabular}{c | c | c }
\hline\noalign{\smallskip}
$\alpha$ & $\frac{\Delta T^{HD}(\alpha,L)}{\Delta T^{LD}(\alpha,L)}$ & $\frac{\Delta T^{NED}(\alpha,L)}{\Delta T^{LD}(\alpha,L)}$  \\
\noalign{\smallskip}\hline\noalign{\smallskip}
$-0.0000075$ & $1.2087000$ & $0.9920636$ \\
$-0.0000100$ & $1.1790380$ & $0.9820383$ \\
$-0.0000125$ & $1.1987310$ & $0.9682116$ \\
$-0.0000150$ & $1.2351010$ & $0.9501191$ \\
$-0.0000175$ & $1.2755030$ & $0.9272487$ \\
$-0.0000200$ & $1.3213920$ & $0.8990362$ \\
$-0.0000225$ & $1.3841210$ & $0.8648598$ \\
$-0.0000250$ & $1.4062240$ & $0.8240357$ \\
$-0.0000275$ & $1.6524540$ & $0.7758114$ \\
\noalign{\smallskip}\hline
\end{tabular}
\end{table}

\setcounter{chapter}{7}
\setcounter{equation}{0} 
\noindent {\bf 7. Concluding remark}
\label{sec:concludings}

In the context of controlled branching processes with random control functions, assuming a general parametric framework for the offspring distribution, we have studied the minimum disparity estimation of its main parameter.

First, we have established conditions for the existence and uniqueness of MDEs for a general discrete model. Moreover, it has been established that the proposed MDEs are strongly consistent as the associated nonparametric estimators are. In particular, we have considered as the nonparametric estimator of the offspring law, the MLE based on the observation of the entire family tree until a certain generation, which is consistent under some regularity conditions. Based on this nonparametric estimator, the limiting normality of the corresponding MDEs of the offspring parameter, suitably normalized, has been also established. These results are regarded as a generalization of those given for Bienaym\'e--Galton--Watson processes (see \citet*{Sriram-2000}), by considering the more general branching structure given by CBPs, and more disparity measures besides Hellinger one.

The MDEs proposed for the offspring parameter are appropriate robust alternatives to the  MLE based on the whole family tree. Focussing our attention on the MHDE and MNEDE, through  a simulated example, we show both are robust against outliers, showing more insensitive the MNEDE to gross-errors at points far from the offspring parameters, and the MHDE when they are at points close to the same one. However, the robustness against inliers is only kept by MNEDE.

\vskip 14pt
\noindent {\large\bf Acknowledgements}

This research has been supported by the Ministerio de Educaci\'on, Cultura y Deporte (grant FPU13/03213), Ministerio de Econom\'ia y Competividad (grant MTM2012-31235), the Gobierno de Extremadura (grant GR15105) and the FEDER.
\par

\setcounter{chapter}{8}
\setcounter{equation}{0} 
\noindent {\bf Appendix}

\noindent{\bf Proof of Theorem \ref{thm:exist-uniq-T}}

\emph{(i)} It is immediate from the definition of $\tilde{\Gamma}_\rho$ and the continuity of $\rho(q,\cdot)$ in $C_\rho$.

\emph{(ii)} From $\inf_{t\in \Theta\backslash K}\rho(p(\theta^*),t)>0$, it is deduced that $\theta^*\in K$, and hence, $\min_{t\in K}\rho(p(\theta^*),t)=0$; consequently, $T^\rho(p(\theta^*))$ exists. Since the function $G(\cdot)$ is nonnegative and has a unique zero at 0, $\rho(p(\theta^*),\theta)=0$ if and only if $p(\theta^*)=p(\theta)$, and from the identifiability of $\F_\Theta$, this can only occurs when $\theta^*=\theta$.

\noindent{\bf Proof of Theorem \ref{thm:contin-T}}

We present an adaptation and extension of the proofs of Proposition 2 in \citet*{Basu-MNEDE} and of Theorem 3.2 in \citet*{Park-Basu-2004}, developed for general continuous models.

Let $\theta=T^\rho(q)$ (there exists and it is unique by Theorem \ref{thm:exist-uniq-T}). For each $t\in\Theta$, using the mean value theorem for the functions $h_k(y)=G(y/p_k(t)-1)$, for each $k\geq 0$, $y>0$, one can prove that $|\rho(q_n,t)-\rho(q,t)|\leq M \sum_{k=0}^\infty |q_{n,k}-q_k|\to 0$, as $n\to\infty$, being $M$ an upper bound of the function $G'(\cdot)$. Hence,
\begin{equation}\label{eq:lim-sup-disparity}
\sup_{t\in\Theta} |\rho(q_n,t)-\rho(q,t)|\to 0,
\end{equation}
obtaining that $\rho(\cdot,t)$ is continuous in $l_1$ for each $t\in\Theta$. From this latter, it is deduced that $q_n\in\tilde{\Gamma}_\rho$ eventually. In fact, if $q_n\notin\tilde{\Gamma}_\rho$ eventually, for all $N\in\N$, there exists $k_N>N$ such that
$$\inf_{t\in\Theta\backslash C_\rho} \rho(q_{k_N},t)\leq\min_{t\in C_\rho} \rho(q_{k_N},t),$$
therefore $q\not\in\tilde{\Gamma}_\rho$, which is in contradiction with the hypotheses of the theorem. Thus, using Theorem \ref{thm:exist-uniq-T}, there exists $T^\rho(q_n)$, which we denote $\theta_n$ to ease the notation, and $\theta_n\in C_\rho$ eventually. Finally, one has to show that $\theta_n\to\theta$.

From \eqref{eq:lim-sup-disparity}, $\rho(q_n,\theta_n)\to \rho(q,\theta)$ and $|\rho(q_n,\theta_n)-\rho(q,\theta_n)|\to 0$ are deduced, so $\rho(q,\theta_n)\to \rho(q,\theta)$.

If the sequence $\{\theta_n\}_{n\geq 0}$ does not converge to $\theta$, then there exists a subsequence $\{\theta_{n_j}\}_{j\in\N}\subseteq\{\theta_{n}\}_{n\in\N}$ such that $\theta_{n_j}\to\theta^*\neq\theta$, as $j\to\infty$. From (A1), taking into account Remark \ref{rem:continuity}~\emph{(i)}, one has that $\rho(q,\cdot)$ is continuous and $\rho(q,\theta_{n_j})\to \rho(q,\theta^*)$, as $j\to\infty$. Due to all of the above, one has $\rho(q,\theta)=\rho(q,\theta^*)$, which contradicts the uniqueness of $T^\rho(q)$.

\noindent{\bf Proof of Theorem \ref{thm:contin-T-HD}}

It is analogous to the previous proof taking into account that \eqref{eq:lim-sup-disparity} for $\rho=HD$ is followed from
$$\sup_{t\in\Theta}|HD(q_n,t)^{1/2}-HD(q,t)^{1/2}|\leq \|q_n^{1/2}-q^{1/2}\|_2.$$

\noindent{\bf Proof of Theorem \ref{thm:consistency-MDE}}\label{ape:consistency-MDE}

First of all, note that since $\tilde{p}_{n,k}\to p_k$ a.s., for each $k\geq 0$, by Glick's Theorem (see \citet*{Devroye-Gyorfi},  p.10), one has $\tilde{p}_{n}\to p$ a.s. in $l_1$.

\emph{(i)} It is immediate from Theorem \ref{thm:contin-T} and the fact that $\tilde{p}_{n}\to p$ a.s. in $l_1$.

\emph{(ii)} The proof is analogous to that of Theorem 3.2 in \citet*{Sriram-2000}. Bearing in mind Theorem \ref{thm:contin-T-HD}, to obtain the eventual existence and the consistency it is enough to prove $||\tilde{p}_n^{1/2}-p^{1/2}||_2\to 0$, a.s. and this is shown from the convergence of $\tilde{p}_{n}$ to $p$ in $l_1$ and the inequality $||\tilde{p}_n^{1/2}-p^{1/2}||_2^2\leq ||\tilde{p}_n-p||_1$.

The measurability of $\tilde{\theta}_n^{\rho}(\tilde{p}_{n})$ and $\tilde{\theta}_n^{HD}(\tilde{p}_{n})$ is obtained by  Corollary 2.1 in \citet*{Brown-Purves}.

\noindent{\bf Proof of Theorem \ref{thm:asymptotic-normality}~\emph{(i)}}

To prove \emph{(i)} we  adapt and extend  the proofs of Theorem 1 in \citet*{Basu-MNEDE} and of Theorem 3.4  in \citet*{Park-Basu-2004} developed for general continuous models. In order to facilitate the proof, we will assume that $P[Z_n\to\infty]=1$.

Let $\dot{\rho}(\hat{p}_n,\theta)$ and $\ddot{\rho}(\hat{p}_n,\theta)$ be the first and the second derivative of $\rho(\hat{p}_n,\theta)$ with respect to $\theta$. Since $\tilde{\theta}_n^\rho(\hat{p}_n)=\arg\min_{\theta\in\Theta}\rho(\hat{p}_n,\theta)$, from the Taylor series expansion of $\dot{\rho}(\hat{p}_n,\tilde{\theta}_n^\rho(\hat{p}_n))$ around $\theta_0$ one obtains
\begin{equation*}
    \Delta_{n-1}^{1/2}(\tilde{\theta}_n^\rho(\hat{p}_n)-\theta_0)= -\Delta_{n-1}^{1/2}\dot{\rho}(\hat{p}_n,\theta_0)\ddot{\rho}(\hat{p}_n,\theta_n^*)^{-1},
\end{equation*}
being $\theta_n^*$ a point between $\theta_0$ and $\tilde{\theta}_n^\rho(\hat{p}_n)$. Consequently, it is enough to prove
\vspace*{-0.5ex}\begin{eqnarray}
  \ddot{\rho}(\hat{p}_n,\theta_n^*) &\xrightarrow{P}& I(\theta_0),\label{eq:conv-denominator-distrib} \\
  -\Delta_{n-1}^{1/2}\dot{\rho}(\hat{p}_n,\theta_0) &\xrightarrow{d}&  N\left(0,I(\theta_0)\right).\label{eq:conv-numerator-distrib}
\end{eqnarray}\vspace*{-0.5ex}
Observe that
{\begin{eqnarray*}
 \dot{\rho}(\hat{p}_n,\theta) &=& -\sum_{k=0}^\infty p'_k(\theta) A(\delta(\hat{p}_n,\theta,k)), \\
 \ddot{\rho}(\hat{p}_n,\theta_n^*) &=& -\sum_{k=0}^\infty p''_k(\theta_n^*) A(\delta(\hat{p}_n,\theta_n^*,k))\\
 &\phantom{=}&+\sum_{k=0}^\infty A'(\delta(\hat{p}_n,\theta_n^*,k))(1+\delta(\hat{p}_n,\theta_n^*,k))u(\theta_n^*,k)^2 p_k(\theta_n^*).
\end{eqnarray*}}
On the one hand, $A(\delta(\hat{p}_n,\theta_n^*,k))\to 0$ a.s. and  $A'(\delta(\hat{p}_n,\theta_n^*,k))\to 1$ a.s.,  using the consistency of $\hat{p}_{n,k}$, for each $k\geq 0$. Therefore, applying the dominated convergence theorem, (A4) and (A5), one has
$$
\sum_{k=0}^\infty p''_k(\theta_0) A(\delta(\hat{p}_n,\theta_n^*,k)) \xrightarrow{P} 0,
$$
\begin{align*}
\sum_{k=0}^\infty A'(\delta(\hat{p}_n,\theta_n^*,k))(1+\delta(\hat{p}_n,\theta_n^*,k&))u(\theta_0,k)^2 p_k(\theta_0)\xrightarrow{P}\\
 &\sum_{k=0}^\infty u(\theta_0,k)^2p_k(\theta_0)=I(\theta_0).\label{eq:sum-2-a}
\end{align*}
Moreover, as $\theta_n^*$ converges to $\theta_0$ in probability, $A(\delta)$ and $A'(\delta)(1+\delta)$ are bounded,  \eqref{eq:cond-conver-prob1} and \eqref{eq:cond-conver-prob2},
\begin{eqnarray*}
 \sum_{k=0}^\infty p''_k(\theta_n^*) A(\delta(\hat{p}_n,\theta_n^*,k)) &\xrightarrow{P}& 0,\\
 \sum_{k=0}^\infty A'(\delta(\hat{p}_n,\theta_n^*,k))(1+\delta(\hat{p}_n,\theta_n^*,k))u(\theta_n^*,k)^2 p_k(\theta_n^*)&\xrightarrow{P}&I(\theta_0),
\end{eqnarray*}
hence, \eqref{eq:conv-denominator-distrib} yields.

In order to establish \eqref{eq:conv-numerator-distrib}, since
{\begin{align*}
    -\Delta_{n-1}^{1/2}\dot{\rho}(\hat{p}_n,\theta_0)&=\Delta_{n-1}^{1/2}\sum_{k=0}^\infty p'_k(\theta_0) \delta(\hat{p}_n,\theta_0,k)\\
    &\phantom{=}+\Delta_{n-1}^{1/2}\sum_{k=0}^\infty p'_k(\theta_0) [A(\delta(\hat{p}_n,\theta_0,k))-\delta(\hat{p}_n,\theta_0,k)],
\end{align*}}
it is sufficient to prove that
{\begin{eqnarray}
  \Delta_{n-1}^{1/2}\sum_{k=0}^\infty p'_k(\theta_0) \delta(\hat{p}_n,\theta_0,k) &\xrightarrow{d}& N(0,I(\theta_0)), \label{eq:conv-normal-distrib} \\
  \Delta_{n-1}^{1/2}\sum_{k=0}^\infty p'_k(\theta_0) [A(\delta(\hat{p}_n,\theta_0,k))-\delta(\hat{p}_n,\theta_0,k)] &\xrightarrow{P}& 0\label{eq:conv-sumat-numer}.
\end{eqnarray}}
Note that due to $\sum_{k=0}^\infty p'_k(\theta_0)=0$ and to (A4)~\emph{(a)}, then
$$\Delta_{n-1}^{1/2}\sum_{k=0}^\infty p'_k(\theta_0) \delta(\hat{p}_n,\theta_0,k)= \Delta_{n-1}^{-1/2}\sum_{i=0}^{n-1}\sum_{j=1}^{\phi_i(Z_i)}u(\theta_0,X_{ij}).$$ Thus
\begin{equation*}
   \Delta_{n-1}^{1/2}\sum_{k=0}^\infty p'_k(\theta_0) \delta(\hat{p}_n,\theta_0,k) \stackrel{d}{=}\Delta_{n-1}^{-1/2}\sum_{i=0}^{\Delta_{n-1}}u(\theta_0,X_{0i}),
\end{equation*}
where $\stackrel{d}{=}$ indicates \emph{equal in distribution}. Now, bearing in mind that $\Delta_n(\tau_m^{n+1} - 1)(\tau_m - 1)^{-1}\to\tau W$ a.s., with $W$ the limit variable introduced in (A3) (this property is deduced by applying Proposition 3.5 in \citet*{art-EM}), using a central limit theorem (see \citet*{D}), one has
\begin{equation*}
  \Delta_{n-1}^{-1/2}I(\theta_0)^{-1/2}\sum_{i=0}^{\Delta_{n-1}}u(\theta_0,X_{0i}) \xrightarrow{d} N(0,1).
\end{equation*}
Consequently \eqref{eq:conv-normal-distrib} holds.

Respect to \eqref{eq:conv-sumat-numer}, applying $|A(t^2-1)-(t^2-1)|\leq B(t-1)^2$ for some $B>0$ (see \citet*{lindsay}, p. 1107), one has
\begin{align}
    \Big|\Delta_{n-1}^{1/2}\sum_{k=0}^\infty &p'_k(\theta_0) [A(\delta(\hat{p}_n,\theta_0,k))-\delta(\hat{p}_n,\theta_0,k)]\Big|\nonumber\\
    &\leq B \Delta_{n-1}^{1/2} \sum_{k=0}^\infty |u(\theta_0,k)| \Big(\hat{p}_{n,k}^{1/2}-p_k(\theta_0)^{1/2}\Big)^2\nonumber\\
    &=B \Delta_{n-1}^{-1/2}\tau_m^{n/2}\sum_{k=0}^\infty A_{n,k}^2,\label{eq:cond-fin}
\end{align}
being
%
{\begin{eqnarray}
  A_{n,k}^2 &=& \tau_ m^{-n/2} \Delta_{n-1} |u(\theta_0,k)| \left(\hat{p}_{n,k}^{1/2}-p_k(\theta_0)^{1/2}\right)^2\nonumber\\
   &=& 2\tau_m^{-n/2} |s_k'(\theta_0)| \Delta_{n-1} s_k(\theta_0)^{-1}\big(\hat{p}_{n,k}^{1/2}-p_k(\theta_0)^{1/2}\big)^2\hspace*{-0.5em},\label{eq:A-kn}
\end{eqnarray}}
with $s(\theta)=p(\theta)^{1/2}$.
Let us demonstrate $\sum_{k=0}^\infty A_{n,k}^2=o_P(1)$\footnote{We write $X_n = o_P (Y_n)$ to mean $P[|X_n|>\epsilon |Y_n|]\to 0$, as $n\to\infty$, for each $\epsilon>0$.}. To this end, we  prove $\lim_{n\to\infty}\sum_{k=0}^\infty E[A_{n,k}^2]=0$. First, taking into account that in \citet*{art-EM} it was proved that
\begin{equation*}\label{eq:asympt-normal-pk}
(p_k(1-p_k))^{-1/2}\Delta_{n-1}^{1/2}(\widehat{p}_{n,k}-p_k)\stackrel{d}{\rightarrow} N(0,1),
\end{equation*}
\vspace*{-1ex}
one has $A_{n,k}^2=o_P(1)$.

\vspace*{1ex}

\noindent Second, denoting $V_i(k)=\sum_{j=1}^{\phi_{i-1}(Z_{i-1})} (I_{\{X_{i-1j}=k\}}-p_k(\theta_0))$ and from the facts that $E[V_i(k)]=0$, $Var [V_i(k)]=p_k(\theta_0)(1-p_k(\theta_0))E[\varepsilon(Z_{i-1})]$ and $\tau_m^{-n}\sum_{i=1}^n\varepsilon(Z_{i-1})$ is a bounded sequence, one obtains that there exists some $K_0\in\R$ such that
\vspace*{-1ex}\begin{eqnarray*}
  E[A_{n,k}^4] &\leq & 4 \tau_m^{-n} |s_k'(\theta_0)|^2 p_k(\theta_0)^{-1} E\bigg[\bigg(\sum_{i=1}^n V_i(k)\bigg)^2\bigg]  \\
   &=& 4 |s_k'(\theta_0)|^2  (1-p_k(\theta_0))  E\bigg[\frac{\sum_{i=1}^n\varepsilon(Z_{i-1})}{\tau_m^{-n}}\bigg]
   \leq K_0 |s_k'(\theta_0)|^2 (1-p_k(\theta_0)).
\end{eqnarray*}\vspace*{-1ex}

Hence, $\sup_n E[A_{n,k}^4]<\infty$ and $\{A_{n,k}\}_{n\in\N}$ is uniformly integrable for each $k\geq 1$. Therefore, since $A_{n,k}^2=o_P(1)$ as $n\to\infty$, we have $A_{n,k}^2$ converges to $0$ in $l_1$, as $n\to\infty$. Moreover, $\sum_{k=0}^\infty E[A_{n,k}^2]\leq K_0^{1/2}\sum_{k=0}^\infty  |s_k'(\theta_0)|<\infty$, as a consequence, by the dominated convergence theorem, $\lim_{n\to\infty}\sum_{k=0}^\infty E[A_{n,k}^2]=0$.

Now, from \eqref{eq:cond-fin}, since $\tau_m^n\Delta_{n-1}^{-1}\to (\tau_m-1)^{-1}\tau W$ a.s., one has \eqref{eq:conv-sumat-numer}.

\noindent{\bf Proof of Theorem \ref{thm:asymptotic-normality}~\emph{(ii)}}

In order to prove Theorem \ref{thm:asymptotic-normality}~\emph{(ii)}, we will make use of a previous result:

\begin{lemma}\label{lema:previous-asymptotic-normality}
Let be a CBP satisfying (A3) and (A6), with $p=p(\theta_0)$ its offspring distribution. Assume (A2), $p\in\hat\Gamma_{HD}$ and $\theta_0\in int(\Theta)$; then
 $$\Delta_{n-1}^{-1/2}\sum_{l=1}^n\sum_{j=1}^{\phi_{l-1}(Z_{l-1})}s_{X_{l-1j}}'(\theta_0)p_{X_{l-1j}}^{-1/2}\stackrel{d}{\rightarrow}  N\left(0,||s'(\theta_0)||_2^2\right),$$
 with respect to the distribution $P[\cdot| Z_n\to\infty]$.
 \end{lemma}

\begin{proof}
To simplify the proof, we assume that $P[Z_n\to\infty]=1$. Let $\beta_l=\sum_{j=1}^{\phi_{l-1}(Z_{l-1})}s_{X_{l-1j}}'(\theta_0)p_{X_{l-1j}}^{-1/2}$ and $\mathcal{G}_l=\sigma(X_{ij},\phi_i(k): j\geq 1, k\geq 0,i=0,\ldots,l-1)$; then $\{\beta_l,\mathcal{G}_l\}_{l\geq 0}$ is a martingale difference and
\begin{equation}\label{eq:var-beta-l}
    E[\beta_l^2|\mathcal{G}_{l-1}]=\varepsilon(Z_{l-1})||s'(\theta_0)||_2^2,\qquad \mbox{ a.s.}
\end{equation}

Moreover,
{\small\begin{align}\label{eq:asympt-distrib-mu}
 \Delta_{n-1}^{-1/2}\sum_{l=1}^n \beta_l &= \left(\frac{Y_{n-1}}{\Delta_{n-1}}\right)^{1/2} \left(\frac{\tau_m^n}{Y_{n-1}}\right)^{1/2} \Bigg[\frac{1}{\tau_ m^{n/2}}= \sum_{l=1}^n \frac{\left((\varepsilon(Z_{l-1})+1)^{1/2}-(\tau_m^{l-1}\tau W)^{1/2}\right)\beta_l}{(\varepsilon(Z_{l-1})+1)^{1/2}}\nonumber \\
&\phantom{=}+\left(\frac{\tau W}{\tau_m -1}\right)^{1/2}(\tau_m -1)\sum_{l=1}^n  \frac{\tau_m^{-(n-l+1)/2}\beta_l}{(\varepsilon(Z_{l-1})+1)^{1/2}}\Bigg],\nonumber
\end{align}}
with $Y_l=\sum_{i=0}^{l}Z_i$, $l=0,\ldots,n$. Because of $\tau_m^{-n}Y_{n-1}\to (\tau_m -1)^{-1} W$ a.s. and $\Delta_{n-1}^{-1}Y_{n-1}\to \tau^{-1}$ a.s. (again deduced by Proposition 3.5 in \citet*{art-EM}), it is enough to prove
\begin{equation}
   \sum_{l=1}^n\frac{\left((\varepsilon(Z_{l-1})+1)^{1/2}-(\tau_m^{l-1}\tau W)^{1/2}\right)\beta_l}{(\varepsilon(Z_{l-1})+1)^{1/2}}
   =o_P (\tau_m^{n/2}),\label{eq:t1-asympt-distrib}
\end{equation}
that is, $\tau_m^{-n/2}\sum_{l=1}^n \left((\varepsilon(Z_{l-1})+1)^{1/2}-(\tau_m^{l-1}\tau W)^{1/2}\right)\beta_l(\varepsilon(Z_{l-1})+1)^{-1/2}$ converges to 0 in probability, and
\begin{equation}\label{eq:t2-asympt-distrib}
  (\tau_m -1)\sum_{l=1}^n  \frac{\tau_m^{-(n-l+1)/2}\beta_l}{(\varepsilon(Z_{l-1})+1)^{1/2}}\xrightarrow{d} N\left(0,||s'(\theta_0)||_2^2\right),
\end{equation}
as $n\to\infty$. The proof follows similar steps to those given in Theorem 2 in \citet*{Sriram}. For \eqref{eq:t1-asympt-distrib}, using Cauchy-Schwarz inequality,
\begin{align*}
\sum_{l=1}^n \left((\varepsilon(Z_{l-1})+1)^{1/2}-(\tau_m^{l-1}\tau W)^{1/2}\right)&\beta_l(\varepsilon(Z_{l-1})+1)^{-1/2}\leq A_n^{1/2}B_n^{1/2},
\end{align*}
%
with
{\begin{eqnarray*}
  A_n &=& \sum_{l=1}^n \tau_m^{(l-1)/2}\left(((\varepsilon(Z_{l-1})+1)\tau_m^{-(l-1)})^{1/2}-(\tau W)^{1/2}\right)^2, \\
  B_n &=& \sum_{l=1}^n \tau_m^{(l-1)/2}(\varepsilon(Z_{l-1})+1)^{-1}\beta_l^2.
\end{eqnarray*}}

On the one hand, since $(\varepsilon(Z_{l-1})+1)\tau_m^{-(l-1)})^{1/2}\to (\tau W)^{1/2}$ a.s., one has $A_n=o(\sum_{l=1}^n \tau_m^{(l-1)/2})=o(\tau_m^{n/2})$.

Moreover, from \eqref{eq:var-beta-l}, $E[B_n]=O(\tau_m^{n/2})$. As a consequence, $A_n=o(\tau_m^{n/2})$ and $B_n=O_P(\tau_m^{n/2})$\footnote{We write $X_n = O_P (Y_n)$ to mean $P[|X_n|>\epsilon |Y_n|]\to M$, as $n\to\infty$, for certain constant $M$ for each $\epsilon>0$.}, and hence, \eqref{eq:t1-asympt-distrib} is proved.

To obtain \eqref{eq:t2-asympt-distrib}, we define $\gamma_{nj}=\beta_{n-j+1}(\varepsilon(Z_{n-j})+1)^{-1/2}$, $j=1,\ldots,n$; then
\begin{align*}\label{eq:t3-asympt-distrib-mu}
   (\tau_m -1)^{1/2}&\sum_{l=1}^n  \frac{\tau_m^{-(n-l+1)/2}\beta_l}{(\varepsilon(Z_{l-1})+1)^{1/2}}
   =(\tau_m -1)^{1/2}\sum_{j=1}^n\frac{\tau_m^{-j/2} \beta_{n-j+1}}{(\varepsilon(Z_{n-j})+1)^{1/2}}\nonumber \\
   &= U_{Jn} + (\tau_m -1)^{1/2}\sum_{j=J+1}^n \tau_m^{-j/2} \gamma_{nj}=U_{nn},
\end{align*}
with $U_{Jn}=(\tau_m -1)^{1/2}\sum_{j=1}^J \tau_m^{-j/2} \gamma_{nj}$, $J=1,\ldots,n$.

For $J\geq 1$ and given $(t_1,\ldots,t_J)\in\R^J$, using analogous arguments to those given in the proof of Theorem 1 in \citet*{Heyde-Brown-71}, we prove that, as $n\to\infty$.
\begin{align*}
E\Big[\exp\Big(i\sum_{j=1}^J t_j \tau_m^{-j/2}\gamma_{nj}\Big)\Big]\rightarrow &\exp\Big(-\frac{1}{2}||s'(\theta_0)||_2^2\sum_{j=1}^J t_j^2 \tau_m^{-j}\Big).
\end{align*}

 As a consequence, following again the steps of the Theorem 2 in \citet*{Sriram}, one can verify that of $U_{Jn}\xrightarrow{d} U_J$, with $U_J$ following a $N(0,(\tau_m -1)||s'(\theta)||_2^2 \sum_{j=1}^J \tau_m^{-j})$. Moreover, for every $n\geq 0$ and $\epsilon>0$, one has $P\left[|U_{Jn}-U_{nn}|>\epsilon\right] \leq \epsilon^{-2}(\tau_m-1)||s'(\theta)||_2^2\sum_{j=J+1}^\infty \tau_m^{-j}$. As a result, there exists a constant $k_0$ such that
$$\limsup_{n\to\infty} P\left[|U_{Jn}-U_{nn}|>\epsilon\right]\leq k_0 \sum_{j=J+1}^\infty \tau_m^{-j}\to 0, \text{ as } J\to\infty.$$

Finally, from Theorem 25.5 in \citet*{Bi-79} and the fact that $U_J\xrightarrow{d} N(0,||s'(\theta)||_2^2)$, as $J\to\infty$, it is verified $U_{nn}\xrightarrow{d} N\left(0,||s'(\theta_0)||_2^2\right)$, as $n\to\infty$, and hence \eqref{eq:t2-asympt-distrib} is obtained.
\end{proof}

\vspace*{3ex}

Once Lemma \ref{lema:previous-asymptotic-normality} is proved, the proof of Theorem \ref{thm:asymptotic-normality}~\emph{(ii)} is analogous to the proof of Theorem 3.4 in \citet*{Sriram-2000}. Recall that $T^{HD}(\hat{p}_n)=\tilde{\theta}_n^{HD}(\hat{p}_n)$ eventually exists, so applying Theorem 3.3 of \citet*{Sriram-2000},
\begin{eqnarray}\label{eq:expr-MHED-asympt}
\Delta_{n-1}^{1/2}(\tilde{\theta}_n^{HD}(\hat{p}_n)&-&\theta_0)=\Delta_{n-1}^{1/2}\big[a_n\sum_{k=0}^\infty s_k'(\theta_0)(\hat{p}_{n,k}^{1/2}-p_k(\theta_0)^{1/2})\nonumber\\
&-&||s'(\theta_0)||_2^{-2}\sum_{k=0}^\infty s_k'(\theta_0)(\hat{p}_{n,k}^{1/2}-p_k(\theta_0)^{1/2})\big]
\end{eqnarray}
where $a_n\to 0$. To obtain the distribution of \eqref{eq:expr-MHED-asympt}, one has to determine the distribution of $\sum_{k=0}^\infty s_k'(\theta_0)(\hat{p}_{n,k}^{1/2}-p_k(\theta_0)^{1/2})$, which, due to $\theta_0=\arg \min _{\theta\in\Theta} ||p_k(\theta)^{1/2}\linebreak-p_k(\theta_0)^{1/2}||_2$, verifies $\sum_{k=0}^\infty s_k'(\theta_0)(\hat{p}_{n,k}^{1/2}-p_k(\theta_0)^{1/2})=\sum_{k=0}^\infty s_k'(\theta_0)\hat{p}_{n,k}^{1/2}$. From the fact that
\begin{eqnarray*}
\sum_{k=0}^\infty s_k'(\theta_0)\hat{p}_{n,k}^{1/2} &=& \frac{1}{2}\sum_{k=0}^\infty s_k'(\theta_0)p_k(\theta_0)^{-1/2}\hat{p}_{n,k}^{1/2}\\
&\phantom{=}&-\frac{1}{2}\sum_{k=0}^\infty s_k'(\theta_0)p_k(\theta_0)^{-1/2}(\hat{p}_{n,k}^{1/2}-p_k(\theta_0)^{1/2})^2\hspace*{-0.25em},
\end{eqnarray*}
and $I(\theta_0)=4||s'(\theta_0)||_2^{2}$, it suffices
\begin{eqnarray}
  \Delta_{n-1}^{1/2}&\sum_{k=0}^\infty& s_k'(\theta_0)p_k(\theta_0)^{-1/2}\hat{p}_{n,k}^{1/2}  \stackrel{d}{\rightarrow} N\left(0,||s'(\theta_0)||_2^{2}\right),\label{eq:expr-MHDE-asympt-t1}\\
  \Delta_{n-1}^{1/2}&\sum_{k=0}^\infty& s_k'(\theta_0)p_k(\theta_0)^{-1/2}(\hat{p}_{n,k}^{1/2}-p_k(\theta_0)^{1/2})^2 \stackrel{P}{\rightarrow} 0.\label{eq:expr-MHDE-asympt-t2}
\end{eqnarray}

Now, from $\Delta_{n-1}^{1/2}\sum_{k=0}^\infty s_k'(\theta_0)p_k(\theta_0)^{-1/2}\hat{p}_{n,k}^{1/2}=\Delta_{n-1}^{-1/2}\sum_{l=1}^n\sum_{j=1}^{\phi_{l-1}(Z_{l-1})}\linebreak s_{X_{l-1j}}'(\theta_0)p_{X_{l-1j}}(\theta_0)^{-1/2}$ and Lemma \ref{lema:previous-asymptotic-normality}, one has \eqref{eq:expr-MHDE-asympt-t1}.


For each $k, n\geq 0$, we define $C_{n,k}=2^{-1}A_{n,k}$, being $A_{n,k}^2$ the random variables introduced in \eqref{eq:A-kn}. Following the same arguments as in the proof of \emph{(i)}, one establishes \eqref{eq:expr-MHDE-asympt-t2} and this completes the proof.

\noindent{\bf Proof of Theorem \ref{thm:influence-curves}}\label{ape:influence-curves}

\emph{(i)} \emph{(i-a)} Let $\theta_L=T(p(\theta,\alpha,L))$. If the sequence $\{\theta_L\}_{L\geq 0}$ does not converge to $\theta$, as $L\to\infty$, then there will exist a subsequence, which we continue denoting $\{\theta_L\}_{L\geq 0}$, such that $\theta_L\to\theta_1\neq\theta$. From the definition of $\theta_L$,
\begin{equation}\label{eq:disparity-thetaL}
    \rho(p(\theta,\alpha,L),\theta_L)\leq \rho(p(\theta,\alpha,L),t),\quad \forall t\in\Theta,
\end{equation}
follows; moreover, applying a generalization of the dominated convergence theorem (see \citet*{Royden}, p.92), one has
\begin{equation}\label{eq:lim-disparity-thetaL}
    \rho(p(\theta,\alpha,L),\theta_L) \to \rho((1-\alpha)p(\theta),\theta_1),\quad \text{ as } L\to\infty;
\end{equation}
as a consequence, from \eqref{eq:disparity-thetaL} and \eqref{eq:lim-disparity-thetaL},
\begin{equation}\label{eq:inequal-lim-disparity-thetaL}
    \rho((1-\alpha)p(\theta),\theta_1)\leq \rho((1-\alpha)p(\theta),t),\quad \forall t\in\Theta.
\end{equation}

On the one hand, since $\rho^*(\alpha,p(\theta),t)=0$ if and only if $t=\theta$, $\rho^*(\alpha,p(\theta),\theta_1)>0=\rho^*(\alpha,p(\theta),\theta)$. On the other hand, due to $\rho((1-\alpha)p(\theta),t)$ is an increasing function of $\rho^*(\alpha,p(\theta),t)$, one obtains $\rho((1-\alpha)p(\theta),\theta_1)>\rho((1-\alpha)p(\theta),\theta)$, which contradicts \eqref{eq:inequal-lim-disparity-thetaL}.

\emph{(ii-b)} The continuity of the function $L\mapsto T(p(\theta,\alpha,L))$ is immediate and the boundedness of the sequence $\{\theta_L\}_{L\geq 0}$ is deduced from its convergence.

\emph{(ii-c)} By definition of $\theta_L$, from the Taylor series expansion of $\dot{\rho} (p(\theta,\alpha,L),\theta_L)$ around $\theta$ one has
\begin{equation*}
    \frac{\theta_L-\theta}{\alpha}=-\frac{\alpha^{-1}\dot{\rho} (p(\theta,\alpha,L),\theta)}{\ddot{\rho}(p(\theta,\alpha,L),\theta_L^*)},
\end{equation*}
with $\theta_L^*$ a point between $\theta$ and $\theta_L$. Consequently, it will be sufficient to prove
\vspace*{-0.5ex}\begin{eqnarray}
  \lim_{\alpha\to 0} \alpha^{-1}\dot{\rho} (p(\theta,\alpha,L),\theta) &=& -u(\theta,L),\label{eq:lim-numer-thetaL}  \\
  \lim_{\alpha\to 0}\ddot{\rho}(p(\theta,\alpha,L),\theta_L^*) &=& I(\theta).\label{eq:lim-denom-thetaL}
\end{eqnarray}\vspace*{-0.5ex}
\noindent With the same arguments as Theorem \ref{thm:asymptotic-normality}~\emph{(i)}, one can prove
\begin{eqnarray*}
 \dot{\rho}(p(\theta,\alpha,L),\theta) &=& -\sum_{k=0}^\infty p'_k(\theta) A(\delta(p(\theta,\alpha,L),\theta,k)), \\
 \ddot{\rho}(p(\theta,\alpha,L),\theta_L^*) &=& -\sum_{k=0}^\infty p''_k(\theta_L^*) A(\delta(p(\theta,\alpha,L),\theta_L^*,k)) \\
 &\phantom{=}&+\sum_{k=0}^\infty A'(\delta(p(\theta,\alpha,L),\theta_L^*,k))u(\theta_L^*,k)^2 p_k(\theta_L^*)\\
 &\phantom{=}&\cdot\ (1+\delta(p(\theta,\alpha,L),\theta_L^*,k)).
\end{eqnarray*}

\noindent Using L'H\^opital's rule and the fact that $\sum_{k=0}^\infty p'_k(\theta)=0$, \eqref{eq:lim-numer-thetaL} is obtained. To show \eqref{eq:lim-denom-thetaL}, we use the convergence dominated theorem in the above expression.

\noindent \emph{(ii)} The proof follows similar steps to Theorem 7 in \citet*{Beran-77} and it is omitted.

\vskip .65cm
\noindent
Department of Mathematics and Instituto de Computaci\'on Cient\'{\i}fica Avanzada (ICCAEx). University of Extremadura, Avda. de Elvas s/n, 06006 Badajoz, Spain
\vskip 2pt
\noindent
E-mail: mvelasco@unex.es
\vskip 2pt
\noindent
Department of Mathematics and Instituto de Computaci\'on Cient\'{\i}fica Avanzada (ICCAEx). University of Extremadura, Avda. de Elvas s/n, 06006 Badajoz, Spain
\vskip 2pt
\noindent
E-mail: cminuesaa@unex.es
\vskip 2pt
\noindent
Department of Mathematics and Instituto de Computaci\'on Cient\'{\i}fica Avanzada (ICCAEx). University of Extremadura, Avda. de Elvas s/n, 06006 Badajoz, Spain
\vskip 2pt
\noindent
E-mail: idelpuerto@unex.es

\end{document}